%% file: main.tex
\documentclass{article}
\usepackage{preprint}
\usepackage{microtype}

\input{math_commands.tex}

\usepackage{colortbl}
\definecolor{shadeblue}{RGB}{130,176,210}
\definecolor{shadeyellow}{RGB}{243,213,110}
\definecolor{accentred}{RGB}{250,127,111}
\definecolor{oraclegray}{RGB}{238,238,238}
\definecolor{palteal}{RGB}{142,207,201}
\definecolor{palorange}{RGB}{255,190,122}
\usepackage{amsthm}
\usepackage{amssymb}
\newtheorem{proposition}{Proposition}
\newtheorem{remark}[proposition]{Remark}
\usepackage{float}
\usepackage{multirow}
\usepackage{wrapfig}
\floatstyle{ruled}
\newfloat{algorithm}{t}{lop}
\floatname{algorithm}{Algorithm}
\usepackage{subcaption}
\usepackage[hidelinks]{hyperref}
\usepackage{url}

\title{\textcolor{palteal}{L}earning \textcolor{palorange}{M}arginal \textcolor{accentred}{V}iewpoint \textcolor{shadeblue}{U}tility for\\Budgeted Evidence Selection in RAG}

\author[1]{Yiyang Duan}
\author[2]{Ziqi Yu}
\author[1]{Tiantong Wu}
\author[1]{Yebo Wu}
\author[1]{Wei Yang Bryan Lim}
\affiliation[1]{Nanyang Technological University}
\affiliation[2]{University of Toronto}

\logos{%
  \logoimage[0.42in]{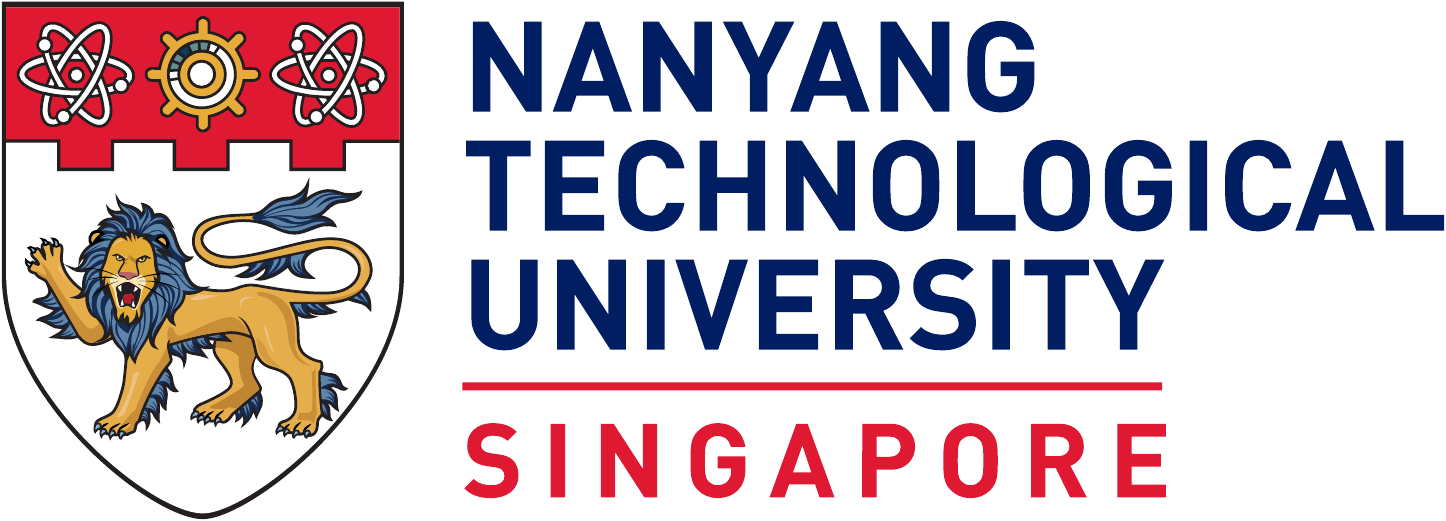}%
  \logodivider[0.42in]%
  \logoimage[0.42in]{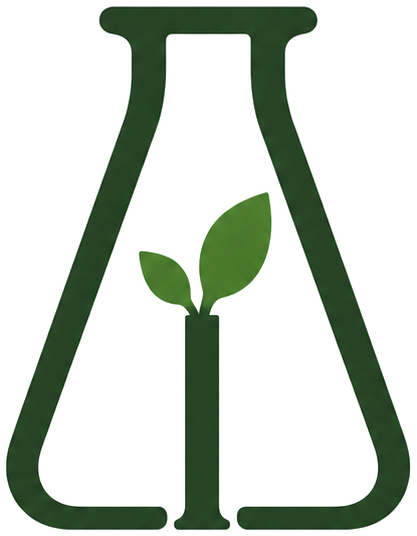}\hspace{0.6em}%
  \logoimage[0.13in]{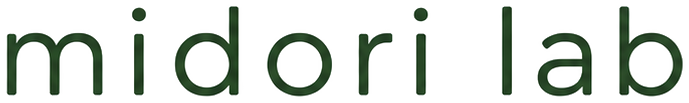}}

\runninghead{Learning Marginal Viewpoint Utility}

\newcommand{\RVC}{\mathrm{RVC}}
\newcommand{\RVCtwo}{\mathrm{RVC}_2}
\newcommand{\RVCn}{\mathrm{RVC}_n}

\begin{document}

\maketitle

\begin{abstract}
Retrieval-augmented generation over multi-viewpoint queries must preserve distinct answers and positions from the retrieved evidence while operating under a fixed context budget. We introduce LMVU, a lightweight ranker that learns to prioritize candidates by the marginal viewpoint coverage they add to the current selection. Training uses viewpoint annotations from each benchmark's development split to compare candidates within the same selection state. At inference, LMVU repeatedly scores budget-feasible candidates as the context grows, without accessing these annotations. We evaluate two forms of recovered viewpoint coverage: $\mathrm{RVC}_1$ assigns each viewpoint full credit for at least one supporting source unit, while $\mathrm{RVC}_2$ assigns half credit for one and full credit for two or more distinct source units. Exact oracles measure achievable coverage under the same evidence pool and budget. Across four viewpoint-annotated benchmarks at a 1{,}024-token budget, LMVU achieves gains of up to $0.047$ and $0.052$ points in $\mathrm{RVC}_1$ and $\mathrm{RVC}_2$, respectively, over the strongest non-oracle baseline for each benchmark and metric, and reaches the exact viewpoint-breadth ceiling on two benchmarks.

\end{abstract}

\input{sections/introduction}

\input{sections/related_work}

\input{sections/problem}

\input{sections/method}

\input{sections/experiments}

\input{sections/discussion_limitations}

\subsection*{AI use statement}
We used AI tools to aid and polish manuscript writing, including improving clarity and readability, and for retrieval and discovery of relevant literature to support the related-work discussion. The authors are responsible for verifying AI-assisted text and checking cited sources, and take full responsibility for the final manuscript, including its claims, results, and references.

\bibliography{references}
\bibliographystyle{preprint}

\appendix

\input{sections/appendix}

\end{document}

%% file: math_commands.tex
\usepackage{amsmath,amsfonts,bm}

\def\eqref#1{equation~\ref{#1}}

\def\1{\bm{1}}

\DeclareMathAlphabet{\mathsfit}{\encodingdefault}{\sfdefault}{m}{sl}
\SetMathAlphabet{\mathsfit}{bold}{\encodingdefault}{\sfdefault}{bx}{n}



%% file: sections/introduction.tex
\section{Introduction}
\label{sec:introduction}
Retrieval-augmented generation (RAG) uses a retriever to find relevant text passages and a generator to produce a response conditioned on this evidence \citep{lewis2020retrieval}. Conventionally, dense passage retrieval ranks passages by their relevance to a query using learned text representations, prioritizing highly ranked passages for inclusion in the context \citep{karpukhin2020dense}.

However, some queries might require coverage of different answers or positions to provide a comprehensive basis for a response. We use \emph{viewpoint} to refer to an answer category or position supported by the evidence. An account of electoral participation, for example, should include three viewpoints: those of people who voted, did not vote, and were ineligible. A context containing only voters' responses represents just one viewpoint, even if every selected record is accurate and relevant. For each query, the available text passages form a candidate pool of evidence items. A limited context budget requires selecting which of these items to include.Highly relevant evidence items may support the same viewpoint while leaving others unrepresented. This leads to a question: 

\begin{center}
\emph{How can we cover different viewpoints within a limited evidence budget?}
\end{center}

A natural way to make better use of the context budget is to favor relevant evidence that adds information beyond what has already been selected. Maximal marginal relevance (MMR) implements this idea by balancing query relevance against similarity to previously selected evidence \citep{carbonell1998use}. This trade-off between relevance and redundancy also informs recent RAG methods that offer high-diversity context selection to improve information recall \citep{wang2025diversity}. However, semantic diversity need not imply viewpoint coverage: differently worded evidence may support the same viewpoint, while similar wording may express different viewpoints.

We introduce Learning Marginal Viewpoint Utility (LMVU), a lightweight linear ranker that selects evidence from a candidate pool to build a context with diverse viewpoints to be supplied to the generator. It ranks candidates by the additional viewpoint coverage they are predicted to contribute to the evidence already selected. During training, we use existing viewpoint annotations from each benchmark's development split to compute which candidate adds more uncovered viewpoints to the same context. The ranker learns these preferences using query relevance, candidate length, and similarities to the selected evidence. Starting from an empty context, LMVU adds the highest-scoring candidate that fits the remaining context budget and updates the scores after each addition.

Figure~\ref{fig:motivation}(a) illustrates a World Values Survey (WVS) query asking how respondents differ in their perceptions of institutional corruption. Under the same candidate pool and 1{,}024-token budget, MMR repeatedly selects evidence supporting three annotated viewpoints, leaving the corruption-free perception unrepresented. LMVU includes evidence for this viewpoint and covers all four.
\begin{figure}[t]
\begin{center}
\includegraphics[width=0.93\linewidth]{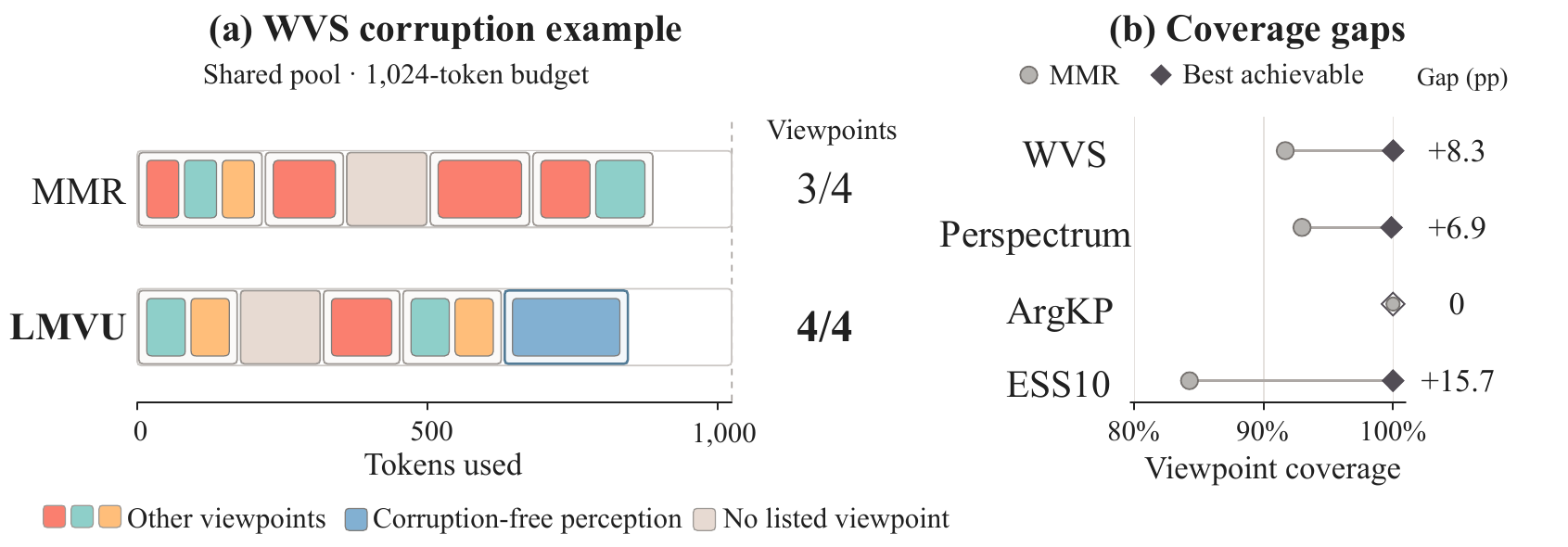}
\end{center}
\caption{(a)~Evidence selected for a WVS corruption-perception query from the same pool under a 1{,}024-token budget. Bars show token use; colors indicate supported viewpoints, including repeated support, and beige denotes no match to the four annotated viewpoints. MMR covers three viewpoints; LMVU also includes the corruption-free perception, covering all four. (b)~Mean viewpoint coverage of MMR and the exact full-pool optimum on each benchmark under the same budget; gaps are in percentage points.}
\label{fig:motivation}
\end{figure}

We evaluate LMVU against relevance-based, diversity-aware, and learned selectors on four benchmarks at a 1{,}024-token budget (Table~\ref{tab:primary}). Figure~\ref{fig:motivation}(b) summarizes MMR's mean coverage gap to the exact optimum under the same pools and budget. LMVU achieves the highest recovered viewpoint coverage (RVC) on WVS and ESS10, reaches the exact optimum on the protected WVS fold, attains 94.6\% coverage on Perspectrum, and matches full coverage on ArgKP. We also measure repeated support using $\RVCtwo$, which gives each viewpoint half credit for one supporting source unit and full credit for two or more distinct units. On ESS10, LMVU outperforms the strongest baseline under both metrics (\S\ref{sec:experiments}).In summary, our contributions are:

\begin{itemize}
    \item We show that budget-constrained evidence selection can leave supported viewpoints unavailable to downstream generation, even when the selected evidence is relevant. Exact coverage oracles reveal recoverable gaps within the same candidate pools, establishing evidence selection as an opportunity to improve viewpoint coverage.
    \item We introduce LMVU, a lightweight linear ranker trained to predict the additional viewpoint coverage a candidate contributes to a growing context. This supervision teaches the ranker how relevance and similarity to selected evidence indicate useful additions, enabling selection to respond to which viewpoints the context already represents.
    \item Across four benchmarks at a 1024-token context budget, LMVU achieves the highest viewpoint coverage on WVS and ESS10, reaches 94.6\% on Perspectrum, and matches full coverage on ArgKP. Further experiments show considering evidence already in the context helps identify useful additions, and that the learned selection rule remains effective on other datasets. Evaluation also shows improvements in repeated support for viewpoints from distinct sources.
\end{itemize}

%% file: sections/related_work.tex
\section{Related Work}
\label{sec:related_work}

\paragraph{Retrieval and diversification.}
Dense passage retrieval scores query--passage relevance \citep{karpukhin2020dense}, while RAG conditions generation on retrieved evidence \citep{lewis2020retrieval}. Maximal marginal relevance \citep{carbonell1998use} and its descendants trade relevance against redundancy with a fixed trade-off; submodular and coverage-based subset selection \citep{lin2011class,mirzasoleiman2015lazier} provide facility-location-style objectives with greedy guarantees. Explicit-aspect diversification re-ranks by aspect coverage \citep{santos2010exploiting}, determinantal point processes sample or greedily select repulsive high-quality subsets \citep{kulesza2012determinantal}, and graph-centrality methods score items by eigenvector prestige in a similarity graph \citep{erkan2004lexrank}. These methods score from fixed similarity or pseudo-aspects; static or state-dependent, they optimize relevance, similarity-based diversity, or proxy aspect coverage rather than supervised marginal gains in annotated viewpoint coverage. LMVU learns an ordering surrogate for benchmark-annotated viewpoint-coverage gains; the two objectives diverge where the budget binds, as \S\ref{sec:experiments} quantifies on four benchmarks.

\paragraph{Adaptive and learned context selection.}
Closest to our setting, several recent methods select evidence under an explicit budget or diversity objective. AdaGReS \citep{peng2025adagres} greedily adds candidates under a redundancy penalty that sums pairwise cosine similarities against the selected evidence, adapting the penalty to the remaining budget. DF-RAG \citep{khan2026dfrag} instead measures diversity as a candidate's geometric distance to the centroid of the selected set and lets an LLM pick the relevance--diversity trade-off per query; we cite it for context but do not evaluate it here. The MMKP selector of \citet{xu2026selfcorrecting} clusters the candidate pool and admits at most one item per cluster, which can prevent representationally close but opposing or complementary evidence from being selected together. JPR \citep{min2021jpr} trains a T5 encoder--decoder for sequential set-aware ranking, and RichRAG \citep{wang2025richrag} trains a 7B encoder--decoder ranker with aspect-coverage supervision (SFT) followed by DPO. LMVU keeps the sequential, state-conditioned view but fits a low-dimensional linear scorer on frozen-embedding statistics with pairwise marginal-coverage labels---far lighter to train than either neural ranker, though bounded by the frozen representation's ability to separate viewpoints. \S\ref{sec:experiments} evaluates AdaGReS, MMKP, JPR, and RichRAG under the same pools, budgets, and protocol.

\paragraph{Learning to rank and data selection.}
Learning-to-rank methods optimize ordering objectives from supervision \citep{liu2009learning}; recent rerankers score candidates conditioned on the query but typically independently of the partially built result set. Data-selection work for instruction tuning and alignment likewise shows that which examples survive selection changes downstream behavior \citep{wei2022finetuned,ouyang2022training,longpre2023data}. LMVU differs in target,  the supervision signal is the \emph{marginal} benchmark-annotated viewpoint coverage a candidate adds \emph{given the current selection state}, and the supervision labels are used only to train the scorer, never at runtime.

\paragraph{Multi-viewpoint and opinion benchmarks.}
Our evaluation spans survey-derived viewpoints (World Values Survey / WorldValuesBench \citep{zhao2024worldvaluesbench}, European Social Survey Round 10 \citep{ess2020round10}), claim-perspective structure (Perspectrum \citep{chen2019seeing}), and argument key-point structure (ArgKP \citep{barhaim2020arguments}). Adjacent opinion- and value-oriented resources \citep{durmus2023measuring,sorensen2023value,kirk2024prism} study whose opinions models express; we study the upstream retrieval decision of which viewpoints enter the context at all, under a label-blindness constraint that mirrors audit practice.

%% file: sections/problem.tex
\section{Problem Formulation}
\label{sec:problem}

\paragraph{Budget-constrained viewpoint coverage.}
For a query $q$, let $E(q)$ be the full prequalified evidence pool and $V_q$ its benchmark-annotated viewpoint inventory. Each item $e$ has token cost $c(e)$ and supports viewpoints $V(e)\subseteq V_q$; we write $V(S)=\bigcup_{e\in S}V(e)$ for the viewpoints covered by a selection. The objective is to select a context $S\subseteq E(q)$ satisfying $\sum_{e\in S}c(e)\le B$ that maximizes \emph{recovered viewpoint coverage}:

$$
\RVC(S)=\frac{\left|\bigcup_{e\in S}V(e)\right|}{|V_q|}.
$$

Each viewpoint receives credit once, regardless of repeated support.

\paragraph{Label-blind selection.}
Gold viewpoint annotations provide training supervision and evaluation targets but are unavailable at inference. The selector must therefore estimate coverage gains using observable evidence features and the current selection state.

\paragraph{Exact oracle and reference headroom.}
With gold access, at a fixed budget $B$ the full-pool oracle solves $S_q^\star\in\arg\max_{S\subseteq E(q)}\RVC(S)$ subject to $\sum_{e\in S}c(e)\le B$; we compute this optimum by viewpoint-subset dynamic programming (Appendix~\ref{app:analysis}). For a reference allocator $\mathcal A$ producing a feasible selection $S_{\mathcal A}(q,B)$, the reference headroom $\mathrm{headroom}_{\mathcal A}(q,B)=\RVC(S_q^\star)-\RVC\!\left(S_{\mathcal A}(q,B)\right)$ is the maximum coverage improvement achievable over $\mathcal A$ under the same evidence pool and budget. Benchmark-level headroom averages these gaps over evaluation queries.

%% file: sections/method.tex
\begingroup
\setlength{\intextsep}{0pt}
\setlength{\abovedisplayskip}{3pt plus 1pt minus 1pt}
\setlength{\belowdisplayskip}{3pt plus 1pt minus 1pt}
\setlength{\belowdisplayshortskip}{3pt plus 1pt minus 1pt}
\section{Method: Learned Marginal Viewpoint Utility}
\label{sec:method}

\subsection{Marginal utility and training states}
The coverage objective of \S\ref{sec:problem} defines exactly the signal a runtime allocator needs that the marginal gain $\Delta V(e \mid q, S)$ of a candidate given the current selection,
\[
\Delta V(e \mid q, S) \;=\; \bigl| V(e) \setminus V(S) \bigr|.
\]
The added coverage of the same candidate changes as the selection grows and falls to zero once its viewpoints are covered; because attainable gains depend on the current selection, training uses within-state comparisons that match the inference-time ranking decision. Computing $\Delta V$ requires gold viewpoint annotations, so the oracle is not deployable, but on development data the realized gains are measurable, and they supervise a surrogate. LMVU (\emph{learned marginal viewpoint utility}) learns to rank candidates by $\Delta V(e \mid q, S)$, a learned scoring function $f_\theta(q, S, e)$ over runtime-visible, label-blind features that serves as an ordering surrogate for marginal coverage gain, and deploys the learned ranking inside a sequential budget-feasible greedy allocator over the full evidence pool $E(q)$ (Figure~\ref{fig:method}).

\begin{figure}[t]
\centering
\includegraphics[width=0.93\textwidth]{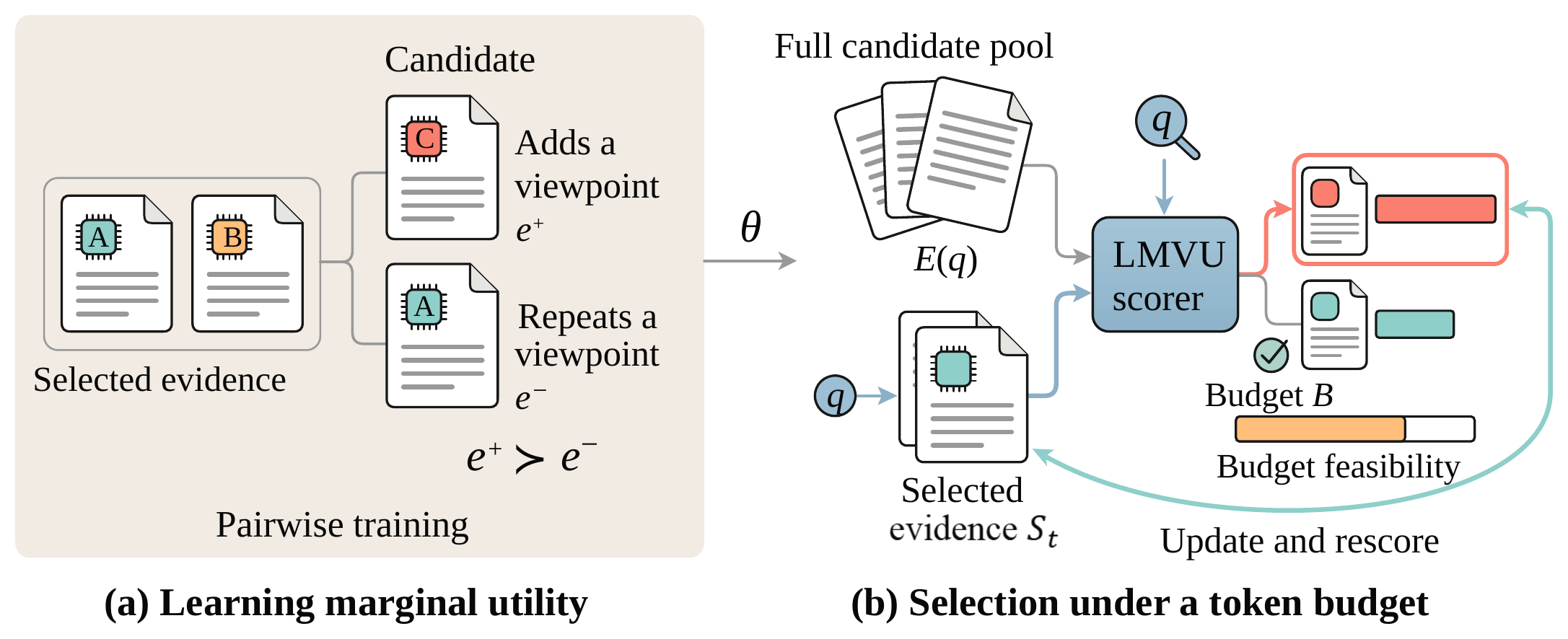}
\caption{LMVU overview. (a)~Learning marginal utility (development, gold visible),with A and B already covered, a candidate adding viewpoint C is the positive $e^{+}$, one repeating a covered viewpoint the negative $e^{-}$; realized gold gains give within-state preferences $e^{+} \succ e^{-}$ for pairwise training, and only the learned $\theta$ leaves this stage. (b)~Selection under a token budget (runtime, label blind),  the scorer ranks budget-feasible candidates from the full pool $E(q)$; the top pick enters $S_t$, the state updates, and candidates are rescored until the budget is exhausted.}
\label{fig:method}
\end{figure}

\paragraph{Training states.}
Development states are harvested on the frozen development pools $P(q)\subseteq E(q)$, $\lambda{=}0.5$ MMR-oversampled, indexed by budget, and used only to harvest training states, by rolling out five development state-visitation policies: dense-relevance greedy, live MMR, two pinned-seed random orders, and a gold-guided marginal-gain policy. Prefixes $(q, S_t)$ of up to nine steps are snapshotted, order-equivalent states are deduplicated per query--budget pair, and each state scores at most twenty label-blind candidates (dense top ten plus a deterministic uniform draw), with gold annotations giving realized gains $\Delta V(e \mid q, S_t)$; unequal-gain candidates within a state supply ranking pairs. Training is confined to the local pools $P(q)$ while the fitted scorer deploys on the full pools $E(q)$; \S\ref{sec:experiments} shows this local training suffices.

\subsection{State-conditioned scoring and learning}
\label{sec:method:model}
\label{sec:method:ranking}
$f_\theta$ is linear in an 18-dimensional feature vector $\phi(q, S, e)$,
\[
f_\theta(q, S, e) \;=\; \theta^\top \tilde\phi(q, S, e),
\]
where $\tilde\phi$ is $\phi$ standardized per dataset on that dataset's development records. The features fall into four groups, candidate descriptors (4), evidence--state geometry (5), atomic--state geometry (5), and shared state descriptors (4), defined in full in Appendix~\ref{app:features} (Table~\ref{tab:features_full}). All eighteen are computable at runtime without gold viewpoint fields, and all similarity geometry reuses the frozen embeddings already computed for the dense scorer, adding no further model dependency. The static ablation M0 restricts $\theta$ to three candidate descriptors: dense relevance, query--candidate cosine, and candidate token length. The four shared state descriptors are constant across candidates within a state, means they cancel in the within-state differences below and cannot alter any within-state argmax; their fitted coefficients are exactly zero in every fitted model. The budget therefore acts on selections through the feasibility filter of \S\ref{sec:method:inference}, not through the learned score.

\paragraph{Pair construction.}
Pointwise regression of $\Delta V$ conflates across-state scale differences, realized gains range from zero to the number of as-yet-uncovered viewpoints, and the attainable scale shrinks as coverage accumulates. We therefore train by \emph{within-state} pairwise ranking. For each development state $(q, S)$ we form ordered pairs $(e^{+}, e^{-})$ with $\Delta V(e^{+} \mid q, S) > \Delta V(e^{-} \mid q, S)$, never pairing candidates across states, capped at 20 pairs per state via deterministic subsample. Tied pairs carry no signal and are excluded. Every retained pair contributes two antisymmetric examples, $(\tilde\phi^{+} - \tilde\phi^{-}, 1)$ and $(\tilde\phi^{-} - \tilde\phi^{+}, 0)$, making the preference direction-consistent and an intercept unnecessary.

\paragraph{Objective.}
The surrogate is fit by a regularized pairwise logistic (Bradley--Terry) loss over standardized symmetric feature differences,
\[
\min_{\theta}\;\; \frac{1}{2}\lVert\theta\rVert^2 \;+\; C \sum_{(q,S)} \sum_{(e^{+},e^{-}) \in \mathcal{P}(q,S)} \log\Bigl(1 + \exp\bigl(-\theta^\top\bigl(\tilde\phi(q,S,e^{+}) - \tilde\phi(q,S,e^{-})\bigr)\bigr)\Bigr),
\]
where $\mathcal{P}(q,S)$ is the state's sampled pair set. Two properties motivate this form. (i) Differencing cancels additive state-specific offsets and the per-dataset standardization mean, leaving only per-feature scale normalized and within-state orderings free of across-state gain scale. (ii) A pairwise loss optimizes exactly the decision the allocator faces, ranking the budget-feasible candidates of one state---and amounts to a concordance (AUC-type) criterion on realized gains. \S\ref{sec:mechanism} reports exactly this criterion on development states (Table~\ref{tab:primary}(b)).  
    The objective is convex,   optimized to convergence with LBFGS ($C{=}1.0$); through the logistic link, the model parameterizes the pairwise preference probability that $e^{+}$ adds more coverage than $e^{-}$ in state $(q,S)$.

\paragraph{Remark.}
MMR's fixed $\lambda$ trade-off uses fixed coefficients to combine relevance with selection-dependent similarity at every step (Remark~\ref{rem:mmr}); in LMVU the learned weights are fixed at inference, but they act on geometry features that update with the selection, so the effective novelty penalty varies with what is already covered, the same candidate can rank highly early and drop once its viewpoint enters the selection (the rotation behavior visible in the case studies of Appendix~\ref{app:cases}).

\subsection{Sequential allocation}
\label{sec:method:inference}
\begin{wrapfigure}[9]{l}{0.51\textwidth}
\setlength{\parskip}{0pt}
\vspace{-6pt}
\hrule height\heavyrulewidth
\vspace{2pt}
\refstepcounter{algorithm}\label{alg:lmvu}%
{\fontsize{8}{9.5}\selectfont\bfseries\noindent Algorithm \thealgorithm~~LMVU budget-feasible greedy allocation\par}
\vspace{1pt}
\hrule height\lightrulewidth
\vspace{3pt}
\centering
\fontsize{8}{9.2}\selectfont
\begin{tabular}{@{}r@{\hspace{0.8em}}l@{}}
\multicolumn{2}{@{}l@{}}{\textbf{Input:} query $q$; full pool $E(q)$; budget $B$; trained scorer $f_\theta$}\\
\multicolumn{2}{@{}l@{}}{\textbf{Output:} context $S \subseteq E(q)$ within budget $B$}\\[2pt]
1 & $S_0 \leftarrow \emptyset$;\; $u_0 \leftarrow 0$;\; $t \leftarrow 0$\\
2 & \textbf{while} true \textbf{do}\\
3 & \hspace{1.2em} $F_t \leftarrow \{\, e \in E(q) \setminus S_t \,:\, u_t + c(e) \le B \,\}$\\
4 & \hspace{1.2em} \textbf{if} $F_t = \emptyset$ \textbf{then break}\\
5 & \hspace{1.2em} $e^{\ast} \leftarrow \arg\max_{e \in F_t} f_\theta(q, S_t, e)$\\
6 & \hspace{1.2em} $S_{t+1} \leftarrow S_t \cup \{e^{\ast}\}$;\; $u_{t+1} \leftarrow u_t + c(e^{\ast})$;\; $t \leftarrow t + 1$\\
7 & \textbf{end while}; \textbf{return} $S_t$\\
\end{tabular}
\par
\vspace{2pt}
\hrule height\heavyrulewidth
\end{wrapfigure}

At inference, the allocator repeats, score budget-feasible candidates with $f_\theta(q, S_t, \cdot)$, add the argmax, update the state, and stop when no candidate fits or the pool is exhausted (Algorithm~\ref{alg:lmvu}, Figure~\ref{fig:method}). Token accounting uses whole-evidence token lengths; partial evidence is never selected, and a candidate is eligible only while its full cost fits the remaining budget, never truncating an item to fit.

\paragraph{Cost and determinism.}
Scoring costs $O(|F_t|\, d)$ per step with $d{=}18$ features; the feature engine maintains incremental sufficient statistics over the selection, running similarity extrema, the selected-embedding sum, per-atom maxima, with no pool-level similarity matrix materialized. The argmax is deterministic; inference runs one deterministic pass per query (Appendix~\ref{app:protocol}).

\paragraph{Properties.}
LMVU inherits the greedy loop's feasibility guarantee but claims no coverage approximation bound: under the knapsack budget even raw-gain greedy carries no constant-factor guarantee, and the learned surrogate is not the true submodular gain (formal statements: Appendix~\ref{app:analysis}). Whether the predicted ordering tracks realized marginal gains is measured directly by the development-fold ranking metrics and learnability analysis (Appendices~\ref{app:dev} and~\ref{app:learn}).

\endgroup

%% file: sections/experiments.tex
\section{Experiments and Analysis}
\label{sec:experiments}

\subsection{Setup}

\paragraph{Benchmarks.}
We evaluate four viewpoint-annotated benchmarks: survey-derived WVS, claim--perspective Perspectrum \citep{chen2019seeing}, argument-key-point ArgKP (a saturation control) \citep{barhaim2020arguments}, and ESS10 survey items \citep{ess2020round10}. Selection uses the \emph{full} pool $E(q)$ (mean sizes $507/9/76/88$ at budget 1024; Appendix~\ref{app:structure}); training-state harvesting alone uses smaller frozen $\lambda{=}0.5$ MMR-oversampled pools ($4\times$ budget multiplier; \S\ref{sec:method}). Budgets are $\{512,1024,2048\}$ (1024 primary). The protected WVS fold allows only predeclared evaluations and read-only audits, never fitting (Appendix~\ref{app:protocol}). Pool and feature geometry uses one frozen local all-MiniLM-L6-v2 encoder \citep{reimers2019sentencebert}.

\paragraph{Metrics.}
Let $n_v(S)$ count distinct source units supporting viewpoint $v$ in $S$, survey respondents in WVS and ESS10, claim-scoped documents in Perspectrum, and individual arguments in ArgKP. For a positive integer $n$, we evaluate
\[
\RVCn(S) \;=\; \frac{1}{|V_q|} \sum_{v \in V_q} \min\bigl\{ n_v(S)/n,\, 1 \bigr\}.
\]
Here $n=1$ recovers $\RVC$ (breadth); we also report $n=2$ ($\RVCtwo$, repeated support), where zero, one, or at least two supporting units contribute $0$, $0.5$, or $1$ per viewpoint. The offline gold-access bound $\mathrm{Oracle}_n(q,B)$ is the maximum $\RVCn(S)$ over selections from the same full pool $E(q)$ satisfying $\sum_{e\in S}c(e)\le B$. Oracle denotes Oracle$_1$; Oracle$_2$ optimizes repeated support separately, so their selected subsets may differ (Appendix~\ref{app:rvc2}).

\paragraph{Baselines.}
We compare Random, BM25, FacLoc, xQuAD, SubMix, LexRank, Dense, CeRerank, Cluster, MMR, static M0, and LMVU (M1), sharing full pools and token accounting. Four published methods are evaluated under the same pools, budgets, and protocol: the allocation heuristics AdaGReS \citep{peng2025adagres} and the MMKP selector of \citet{xu2026selfcorrecting}, and the training-heavy neural rerankers JPR \citep{min2021jpr} and RichRAG \citep{wang2025richrag}, whose rows are shaded separately in Table~\ref{tab:primary}. The predeclared kDPP comparator appears in the appendix budget grid (Table~\ref{tab:rvc2_full}); definitions and five predeclared baseline locks are in Appendix~\ref{app:protocol}.
For development-state ranking, we report normalized discounted cumulative gain at rank five (NDCG@5), using realized marginal coverage gains as relevance labels, and pairwise concordance with the true gain ordering.

\paragraph{Protocol.}
Each benchmark follows adapter verification $\rightarrow$ frozen pools $\rightarrow$ grouped development training and lock $\rightarrow$ held-out evaluation. Audits found zero runtime gold-field or provider accesses (Appendix~\ref{app:protocol}). The LF matrix and full-pool baseline rerun contain 28{,}314 per-query allocator runs. Queries are the evaluation units: 6, 97, 6, and 34 for WVS, Perspectrum, ArgKP, and ESS10, respectively. WVS uses protected respondent fold USA~4 over the same queries; transfer benchmarks use development/held-out query splits. Deterministic allocators contribute one trajectory per query. Ten seeds vary only Random, xQuAD, and Cluster; learned rows (M0, JPR, RichRAG, LMVU) vary only through three training-seed refits, with deterministic inference. Given prior held-out inspection, this full-pool extension is exploratory rather than a fresh confirmatory test.

\subsection{Coverage under fixed budgets}
\label{sec:stats}

\begin{table}[!t]
\caption{Main results at budget 1024 on the full per-query pools.}
\label{tab:primary}
\centering
\captionsetup[subtable]{font=small,skip=2pt,position=top}
\begin{subtable}[t]{\linewidth}
\centering
\caption{$\RVC$ and $\RVCtwo$ with metric-specific exact full-pool oracle bounds (`--' marks the inapplicable half of each oracle row). Avg is the unweighted mean across the four benchmarks. Within each column, best is \textbf{bold} and second best is \underline{underlined} (oracle bounds excluded); shading encodes magnitude on separate scales for each column. Yellow rows are the training-heavy neural baselines (JPR, RichRAG; lock v5 with a top-64 dense prefilter, Appendix~\ref{app:protocol}); they are excluded from the bold/underline comparison and from the per-column shading scales.}
\label{tab:primary-results}
\footnotesize
\setlength{\tabcolsep}{2.4pt}
\renewcommand{\arraystretch}{1.08}
\begin{tabular}{lcccccccccc}
\toprule
& \multicolumn{5}{c}{$\RVC$ (breadth)} & \multicolumn{5}{c}{$\RVCtwo$ (repeated support)} \\
\cmidrule(lr){2-6}\cmidrule(lr){7-11}
Method & WVS & Perspectrum & ArgKP & ESS10 & Avg & WVS & Perspectrum & ArgKP & ESS10 & Avg \\
\midrule
Random & \cellcolor{shadeblue!24}0.9088 & \cellcolor{shadeblue!24}0.9111 & \cellcolor{shadeblue!24}0.9933 & \cellcolor{shadeblue!28}0.8416 & \cellcolor{shadeblue!24}0.9137 & \cellcolor{shadeblue!28}0.8757 & \cellcolor{shadeblue!24}0.7932 & \cellcolor{shadeblue!20}0.9720 & \cellcolor{shadeblue!24}0.7223 & \cellcolor{shadeblue!24}0.8408 \\
BM25 & \cellcolor{shadeblue!8}0.7000 & \cellcolor{shadeblue!20}0.8859 & \cellcolor{shadeblue!28}\underline{1.0000} & \cellcolor{shadeblue!24}0.7858 & \cellcolor{shadeblue!16}0.8429 & \cellcolor{shadeblue!16}0.6583 & \cellcolor{shadeblue!24}0.7720 & \cellcolor{shadeblue!24}0.9881 & \cellcolor{shadeblue!24}0.6769 & \cellcolor{shadeblue!20}0.7738 \\
FacLoc & \cellcolor{shadeblue!24}\underline{0.9167} & \cellcolor{shadeblue!24}0.9108 & \cellcolor{shadeblue!28}\underline{1.0000} & \cellcolor{shadeblue!24}0.8130 & \cellcolor{shadeblue!24}0.9101 & \cellcolor{shadeblue!28}\textbf{0.9167} & \cellcolor{shadeblue!24}0.8045 & \cellcolor{shadeblue!24}0.9881 & \cellcolor{shadeblue!24}0.6781 & \cellcolor{shadeblue!24}\underline{0.8469} \\
xQuAD & \cellcolor{shadeblue!12}0.7700 & \cellcolor{shadeblue!24}0.9239 & \cellcolor{shadeblue!28}\underline{1.0000} & \cellcolor{shadeblue!24}0.7562 & \cellcolor{shadeblue!20}0.8625 & \cellcolor{shadeblue!20}0.7114 & \cellcolor{shadeblue!24}0.8083 & \cellcolor{shadeblue!20}0.9733 & \cellcolor{shadeblue!24}0.6422 & \cellcolor{shadeblue!20}0.7838 \\
SubMix & \cellcolor{shadeblue!16}0.7833 & \cellcolor{shadeblue!24}0.9139 & \cellcolor{shadeblue!28}\underline{1.0000} & \cellcolor{shadeblue!24}0.7343 & \cellcolor{shadeblue!20}0.8579 & \cellcolor{shadeblue!20}0.6898 & \cellcolor{shadeblue!24}0.8002 & \cellcolor{shadeblue!20}0.9714 & \cellcolor{shadeblue!24}0.6173 & \cellcolor{shadeblue!20}0.7697 \\
LexRank & \cellcolor{shadeblue!16}0.7833 & \cellcolor{shadeblue!24}0.9122 & \cellcolor{shadeblue!28}\underline{1.0000} & \cellcolor{shadeblue!24}0.7314 & \cellcolor{shadeblue!20}0.8567 & \cellcolor{shadeblue!20}0.7315 & \cellcolor{shadeblue!24}0.7986 & \cellcolor{shadeblue!24}0.9881 & \cellcolor{shadeblue!24}0.6220 & \cellcolor{shadeblue!20}0.7850 \\
AdaGReS & \cellcolor{shadeblue!8}0.7167 & \cellcolor{shadeblue!4}0.7509 & \cellcolor{shadeblue!20}0.9762 & \cellcolor{shadeblue!20}0.6647 & \cellcolor{shadeblue!12}0.7771 & \cellcolor{shadeblue!16}0.6148 & \cellcolor{shadeblue!4}0.5018 & \cellcolor{shadeblue!4}0.9143 & \cellcolor{shadeblue!20}0.5138 & \cellcolor{shadeblue!12}0.6362 \\
MMKP & \cellcolor{shadeblue!4}0.6241 & \cellcolor{shadeblue!28}\underline{0.9482} & \cellcolor{shadeblue!16}0.9667 & \cellcolor{shadeblue!4}0.2502 & \cellcolor{shadeblue!4}0.6973 & \cellcolor{shadeblue!4}0.3120 & \cellcolor{shadeblue!28}0.8329 & \cellcolor{shadeblue!4}0.9173 & \cellcolor{shadeblue!4}0.1251 & \cellcolor{shadeblue!4}0.5468 \\
\midrule
Dense & \cellcolor{shadeblue!16}0.7833 & \cellcolor{shadeblue!24}0.9126 & \cellcolor{shadeblue!28}\underline{1.0000} & \cellcolor{shadeblue!24}0.7306 & \cellcolor{shadeblue!20}0.8566 & \cellcolor{shadeblue!20}0.6898 & \cellcolor{shadeblue!24}0.8014 & \cellcolor{shadeblue!20}0.9714 & \cellcolor{shadeblue!20}0.6154 & \cellcolor{shadeblue!20}0.7695 \\
CeRerank & \cellcolor{shadeblue!12}0.7380 & \cellcolor{shadeblue!24}0.9027 & \cellcolor{shadeblue!28}\underline{1.0000} & \cellcolor{shadeblue!20}0.6239 & \cellcolor{shadeblue!16}0.8161 & \cellcolor{shadeblue!20}0.6704 & \cellcolor{shadeblue!24}0.7923 & \cellcolor{shadeblue!12}0.9458 & \cellcolor{shadeblue!20}0.5285 & \cellcolor{shadeblue!16}0.7343 \\
Cluster & \cellcolor{shadeblue!20}0.8833 & \cellcolor{shadeblue!24}0.9135 & \cellcolor{shadeblue!20}0.9752 & \cellcolor{shadeblue!24}0.8134 & \cellcolor{shadeblue!24}0.8964 & \cellcolor{shadeblue!24}0.7753 & \cellcolor{shadeblue!24}0.7999 & \cellcolor{shadeblue!16}0.9631 & \cellcolor{shadeblue!24}0.6960 & \cellcolor{shadeblue!24}0.8086 \\
MMR & \cellcolor{shadeblue!24}\underline{0.9167} & \cellcolor{shadeblue!24}0.9296 & \cellcolor{shadeblue!28}\underline{1.0000} & \cellcolor{shadeblue!28}\underline{0.8427} & \cellcolor{shadeblue!24}\underline{0.9222} & \cellcolor{shadeblue!24}0.7824 & \cellcolor{shadeblue!24}0.8160 & \cellcolor{shadeblue!28}\underline{1.0000} & \cellcolor{shadeblue!28}\underline{0.7273} & \cellcolor{shadeblue!24}0.8314 \\
M0 & \cellcolor{shadeblue!16}0.7981 & \cellcolor{shadeblue!28}\textbf{0.9495} & \cellcolor{shadeblue!4}0.9250 & \cellcolor{shadeblue!24}0.8287 & \cellcolor{shadeblue!20}0.8753 & \cellcolor{shadeblue!20}0.7444 & \cellcolor{shadeblue!28}\textbf{0.8480} & \cellcolor{shadeblue!4}0.9131 & \cellcolor{shadeblue!24}0.7206 & \cellcolor{shadeblue!24}0.8065 \\
\cellcolor{accentred!18}LMVU & \cellcolor{shadeblue!28}\textbf{1.0000} & \cellcolor{shadeblue!28}0.9457 & \cellcolor{shadeblue!28}\textbf{1.0000} & \cellcolor{shadeblue!28}\textbf{0.8893} & \cellcolor{shadeblue!28}\textbf{0.9588} & \cellcolor{shadeblue!28}\underline{0.9015} & \cellcolor{shadeblue!28}\underline{0.8373} & \cellcolor{shadeblue!28}\textbf{1.0000} & \cellcolor{shadeblue!28}\textbf{0.7790} & \cellcolor{shadeblue!28}\textbf{0.8795} \\
\midrule
JPR & \cellcolor{shadeyellow!16}0.8361 & \cellcolor{shadeyellow!24}0.9074 & \cellcolor{shadeyellow!28}1.0000 & \cellcolor{shadeyellow!20}0.7197 & \cellcolor{shadeyellow!20}0.8658 & \cellcolor{shadeyellow!24}0.8097 & \cellcolor{shadeyellow!24}0.7996 & \cellcolor{shadeyellow!28}1.0000 & \cellcolor{shadeyellow!24}0.6839 & \cellcolor{shadeyellow!24}0.8233 \\
RichRAG & \cellcolor{shadeyellow!16}0.8167 & \cellcolor{shadeyellow!20}0.8908 & \cellcolor{shadeyellow!4}0.9234 & \cellcolor{shadeyellow!28}0.9569 & \cellcolor{shadeyellow!24}0.8970 & \cellcolor{shadeyellow!24}0.7994 & \cellcolor{shadeyellow!24}0.7906 & \cellcolor{shadeyellow!8}0.9234 & \cellcolor{shadeyellow!28}0.8315 & \cellcolor{shadeyellow!24}0.8362 \\
\midrule
Oracle & \cellcolor{oraclegray}1.0000 & \cellcolor{oraclegray}0.9987 & \cellcolor{oraclegray}1.0000 & \cellcolor{oraclegray}1.0000 & \cellcolor{oraclegray}0.9997 & \cellcolor{oraclegray}-- & \cellcolor{oraclegray}-- & \cellcolor{oraclegray}-- & \cellcolor{oraclegray}-- & \cellcolor{oraclegray}-- \\
Oracle$_2$ & \cellcolor{oraclegray}-- & \cellcolor{oraclegray}-- & \cellcolor{oraclegray}-- & \cellcolor{oraclegray}-- & \cellcolor{oraclegray}-- & \cellcolor{oraclegray}1.0000 & \cellcolor{oraclegray}0.8898 & \cellcolor{oraclegray}1.0000 & \cellcolor{oraclegray}0.9334 & \cellcolor{oraclegray}0.9558 \\
\bottomrule
\end{tabular}
\end{subtable}
\par\vspace{2pt}
\begin{subtable}[t]{\linewidth}
\centering
\caption{Development-only ranking (left) and leave-one-benchmark-out transfer at budget 1024 (right). Ranking scores average seeds and, for WVS, three folds; $\Delta$ is relative to MMR. Best per triple in \textbf{bold}; shading indicates magnitude.}
\label{tab:primary-mech}
\footnotesize
\setlength{\tabcolsep}{2.4pt}
\renewcommand{\arraystretch}{1.08}
\begin{tabular}{lcccccccc}
\toprule
& \multicolumn{3}{c}{dev within-state NDCG@5} & \multicolumn{3}{c}{dev pairwise concordance} & \multicolumn{2}{c}{LOBO transfer (dev, 1024)} \\
\cmidrule(lr){2-4}\cmidrule(lr){5-7}\cmidrule(lr){8-9}
Benchmark & \cellcolor{accentred!18}LMVU & M0$^+$ & M1-SHUF & \cellcolor{accentred!18}LMVU & M0$^+$ & M1-SHUF & RVC & $\Delta$ vs.\ MMR \\
\midrule
WVS & \cellcolor{shadeblue!28}\textbf{0.7944} & \cellcolor{shadeblue!20}0.6326 & \cellcolor{shadeblue!16}0.6002 & \cellcolor{shadeblue!28}\textbf{0.8260} & \cellcolor{shadeblue!28}0.7184 & \cellcolor{shadeblue!28}0.6892 & \cellcolor{shadeblue!28}1.0000 & \cellcolor{accentred!28}+0.1500 \\
Perspectrum & \cellcolor{shadeblue!28}0.7927 & \cellcolor{shadeblue!28}\textbf{0.8004} & \cellcolor{shadeblue!28}0.7960 & \cellcolor{shadeblue!4}\textbf{0.5829} & \cellcolor{shadeblue!4}0.5375 & \cellcolor{shadeblue!4}0.5328 & \cellcolor{shadeblue!20}0.9510 & \cellcolor{accentred!16}+0.0443 \\
ArgKP & \cellcolor{shadeblue!4}\textbf{0.4629} & \cellcolor{shadeblue!4}0.3539 & \cellcolor{shadeblue!4}0.3496 & \cellcolor{shadeblue!8}\textbf{0.6380} & \cellcolor{shadeblue!4}0.5266 & \cellcolor{shadeblue!4}0.5243 & \cellcolor{shadeblue!28}0.9885 & 0.0000 \\
ESS10 & \cellcolor{shadeblue!12}\textbf{0.5715} & \cellcolor{shadeblue!16}0.5498 & \cellcolor{shadeblue!16}0.5502 & \cellcolor{shadeblue!8}\textbf{0.6098} & \cellcolor{shadeblue!8}0.5733 & \cellcolor{shadeblue!12}0.5736 & \cellcolor{shadeblue!4}0.8452 & \cellcolor{accentred!16}+0.0374 \\
\bottomrule
\end{tabular}
\end{subtable}
\par\vspace{2pt}
\begin{subcaptionblock}[t]{\linewidth}
\centering
\caption{Paired LMVU$-$comparator differences on held-out ESS10, means and 95\% bootstrap CIs; same order in both halves on a shared scale.}
\label{tab:primary-forest}
\includegraphics[width=\linewidth]{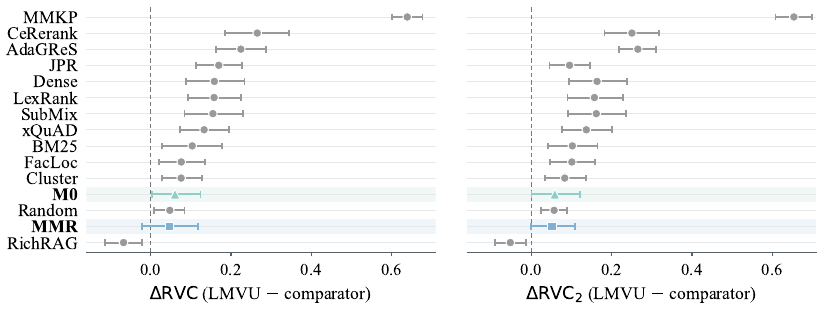}
\end{subcaptionblock}
\end{table}

Table~\ref{tab:primary} is the central result. At budget 1024 the MMR reference allocator leaves exact full-pool oracle headroom of $0.083$ (WVS), $0.069$ (Perspectrum), and $0.157$ (ESS10) RVC points, while ArgKP's headroom is zero (Table~\ref{tab:structure}). LMVU recovers \emph{all} of the MMR-to-oracle headroom on the protected WVS fold---matching the full-pool oracle and gold viewpoint retrieval exactly ($1.0000$)---and improves over MMR on both headroom-positive transfer benchmarks, with the largest gain on the headroom-richest one (ESS10, $+0.0466$ over MMR, headroom recovery $0.30$). On ArgKP, where headroom is zero, MMR already equals the oracle and LMVU ties it exactly; under $\RVC$ every other non-random method except M0 ($0.9250$), Cluster ($0.9752$), AdaGReS ($0.9762$), MMKP ($0.9667$), and RichRAG ($0.9234$) also reaches the ceiling. The Avg columns make the cross-dataset comparison explicit: LMVU attains the best average coverage of all sixteen non-oracle methods on \emph{both} metrics ($0.9588$ $\RVC$, $0.8795$ $\RVCtwo$), ahead of the strongest heuristic rows (MMR $0.9222$, FacLoc $0.8469$) and of both neural rankers (at most $0.8970$ and $0.8362$)---no per-benchmark specialist matches it across datasets.

Three observations stand out, (i) lexical relevance alone (BM25) trails embedding methods on WVS and Perspectrum, and a relevance-free coverage objective (FacLoc) underperforms even random on ESS10; (ii) diversity-structured heuristics are benchmark-dependent---on WVS, FacLoc matches MMR under $\RVC$ and edges past LMVU under $\RVCtwo$ ($0.9167$ vs.\ $0.9015$), and on Perspectrum the cluster-constrained MMKP tops MMR ($0.9482$ vs.\ $0.9296$) yet collapses on ESS10 ($0.2502$ at every budget, its per-cluster cap binding); on ESS10, where pools are smaller and redundancy subtler, no non-learned baseline exceeds MMR; (iii) only the learned state-dependent utility improves over MMR on all three headroom-positive benchmarks.

Table~\ref{tab:primary}(b) collects development-side mechanism and transfer evidence (\S\ref{sec:mechanism}, \S\ref{sec:transfer}); panel (c) shows the ESS10 comparator margins. Paired against MMR at budget 1024, LMVU's differences are nonnegative on every benchmark (ArgKP ties at the ceiling), and the Perspectrum repeated-support gain is the one comparison whose bootstrap interval excludes zero ($+0.0213$, 95\% CI $[+0.007, +0.037]$); the full paired statistics are in Table~\ref{tab:rvc2_budgets}. Against other baselines, larger folds yield significance---on ESS10, LMVU beats MMKP ($+0.639$), CeRerank ($+0.265$), AdaGReS ($+0.225$), JPR ($+0.170$), and LexRank at $p<0.001$, Dense ($p=0.002$), SubMix ($p=0.003$), xQuAD ($p=0.004$), Random ($p=0.005$), FacLoc and Cluster ($p=0.019$), and BM25 ($p=0.035$), while RichRAG reverses the comparison ($-0.068$, $p=0.004$); on Perspectrum it beats AdaGReS, Random, BM25, Dense, SubMix, and Cluster at $p\le0.003$, CeRerank and RichRAG ($p=0.006$), xQuAD ($p=0.007$), JPR ($p=0.024$), and LexRank ($p=0.029$), while the M0 ($-0.004$) and MMKP ($-0.003$) exceptions persist. Development folds match this ordering (Appendix~\ref{app:dev}).

\paragraph{Training-heavy neural baselines.}
The yellow rows of Table~\ref{tab:primary} evaluate two published neural rerankers under the same pools, budgets, and seed discipline (lock v5, Appendix~\ref{app:protocol}); both score candidates with a neural forward pass at every greedy step, restricted to a top-64 dense-relevance prefilter. JPR \citep{min2021jpr}, a T5 listwise ranker, is an ArgKP specialist: it reaches the ceiling already at 512 tokens ($1.0000$, matched only by the appendix's kDPP comparator; LMVU $0.9762$, MMR $0.9667$) and ties it at the larger budgets, but trails MMR on WVS ($0.8361$), Perspectrum ($0.9074$), and ESS10 ($0.7197$). RichRAG \citep{wang2025richrag}, a 7B reranker trained with aspect-coverage SFT and preference DPO, leads ESS10 breadth at every budget ($0.834/0.957/0.988$ against LMVU's $0.729/0.889/0.973$) and ESS10 repeated support at 1024 and 2048 ($0.832$ and $0.967$ against $0.779$ and $0.934$), but trails MMR on WVS and Perspectrum, and its ESS10 repeated support at 512 tokens stays below LMVU's ($0.544$ vs.\ $0.553$). Neither neural method improves over MMR on all three headroom-positive benchmarks, so observation (iii) survives their inclusion. Both also trail LMVU on the cross-benchmark averages (JPR $0.8658$/$0.8233$, RichRAG $0.8970$/$0.8362$, against $0.9588$/$0.8795$ for $\RVC$/$\RVCtwo$): their single-benchmark specialist wins do not compose into cross-dataset leadership. Their training and inference costs sit orders of magnitude above LMVU's: per fit, JPR required $150.5$\,s and RichRAG $962.7$\,s (SFT $+$ DPO) on one L20 GPU---about $2.5$ and $5.6$ GPU-hours for the 18-fit seed--fold grid, plus $\approx11$ GPU-hours of prefiltered neural inference for the locked evaluation batch---whereas LMVU refits its 18 coefficients on CPU in seconds (measured $1.3$\,s for the largest benchmark's 13{,}952 development states) and scores each candidate with a single 18-dimensional dot product at inference. LMVU thus pairs the best average coverage across benchmarks with by far the lightest training and inference footprint.

\subsection{What does state conditioning contribute?}
\label{sec:mechanism}

Table~\ref{tab:primary} contains the key ablation, M0, the static 3-feature scorer. State-dependence drives the three largest gains, LMVU beats M0 by $+0.202$ (WVS), $+0.075$ (ArgKP, where M0 drops below the otherwise universal ceiling at $0.9250$), and $+0.061$ (ESS10). Perspectrum is the exception where M0 exceeds LMVU by $0.004$, a benchmark-level boundary we analyze in \S\ref{sec:discussion}. Dense top-$k$ never exceeds MMR (tying only at the ArgKP ceiling) and stays far from the ceiling where headroom exists (ESS10 $0.731$ vs.\ $1.000$ for the exact oracle), relevance alone does not solve the allocation.

Two development diagnostics complement these held-out results. In the pre-registered learnability probe (Appendix~\ref{app:learn}), within-state ranking on development states puts M1 at NDCG@5 $0.6575$, ahead of the MMR-style heuristic ($0.6313$) and M0 ($0.5869$). Removing the six atomic-granularity features collapses WVS development ranking (NDCG@5 $0.6597 \rightarrow 0.4418$; top-1 positive-gain rate $0.5425 \rightarrow 0.4088$), leaving Perspectrum and ArgKP unchanged (Appendix~\ref{app:ablations}). Both are development-fold diagnostics that localize the state signal's source; neither is held-out evidence.

Table~\ref{tab:primary}(b) (left) adds two predeclared controls on the same development states. A capacity-matched null that keeps all 18 features but shuffles the 14 state features within each state (M1-SHUF) collapses ranking to the static level (WVS $0.6002$ vs.\ LMVU's $0.7944$; ArgKP $0.3496$ vs.\ $0.4629$), so the ranking gain comes from state information rather than capacity. Conversely, augmenting the static scorer with a fourth candidate descriptor (M0$^+$) leaves its selections identical to M0's on Perspectrum, ArgKP, and ESS10 and its WVS breadth at 1024 unchanged ($0.7981$), while its development ranking stays far below LMVU's exactly where the held-out gaps are large. Perspectrum is again the exception, both static variants slightly outrank LMVU , mirroring the held-out $-0.004$.

\subsection{Cross-benchmark transfer}
\label{sec:transfer}

If the learned coefficients capture a reusable coverage regularity rather than benchmark-specific annotation patterns, a model trained without any data from a benchmark should still allocate well on it. Table~\ref{tab:primary}(b) (right) reports this leave-one-benchmark-out probe on development folds at budget 1024, one model trained on the pooled development states of the other three benchmarks improves over MMR on WVS ($+0.150$), Perspectrum ($+0.044$), and ESS10 ($+0.037$), ties at the ArgKP ceiling, and attains full breadth on WVS ($1.0000$). All four cells use development folds only, so the probe leaves the held-out firewall intact; the WVS cell in particular never touches the protected fold.

\subsection{Breadth and repeated support}

Under repeated-support coverage, the ordering largely carries over (M0 leads Perspectrum, LMVU leads ESS10, ArgKP remains tied at the ceiling), with one inversio that on WVS the relevance-free FacLoc edges LMVU ($0.9167$ vs.\ $0.9015$), though both sit well below the depth oracle. LMVU's gains over MMR remain positive on all headroom-positive benchmarks ($+0.119$, $+0.021$, $+0.052$). Depth also holds the larger residual headroom,headroom$_2$ is $0.218$ (WVS), $0.074$ (Perspectrum), and $0.206$ (ESS10), against $0.083$, $0.069$, and $0.157$ for breadth (Table~\ref{tab:structure}). The residual gaps differ, at budget 1024, LMVU's remaining distance to the full-pool oracle on WVS is $0.0000$ in breadth but $0.0985$ in depth---breadth is at the pool optimum and only repeated support is missing---while ESS10 retains a $0.1107$ breadth and a $0.1544$ depth gap.

\subsection{Robustness across embedding models}

\begin{table}[t]
\caption{Sensitivity across encoder-specific pipelines at budget 1024 in the same scopes as Table~\ref{tab:primary}. Each encoder column rebuilds embeddings, pools, and features together (not an encoder-only controlled comparison); the bge columns retain the earlier local-selection protocol (smaller pools; Appendix~\ref{app:protocol}). Left: exact-oracle headroom; right: paired LMVU--MMR $\Delta\RVC$ with 95\% CI (WVS MiniLM-L6, range across the three issue-group deltas). ArgKP is omitted because MMR and LMVU both reach the oracle ceiling under every encoder, giving zero headroom and zero $\Delta\RVC$.}
\label{tab:crossencoder}
\centering
\small
\newcommand{\dcell}[3]{\begin{tabular}{@{}c@{}}#1\\[0.5pt]{[#2, #3]}\end{tabular}}
\setlength{\tabcolsep}{2pt}
\renewcommand{\arraystretch}{1.05}
\begin{tabular}{@{}lcccccc@{}}
\toprule
& \multicolumn{3}{c}{headroom (oracle $-$ MMR)} & \multicolumn{3}{c}{LMVU $-$ MMR $\Delta\RVC$ [95\% CI]} \\
\cmidrule(lr){2-4}\cmidrule(lr){5-7}
& MiniLM-L6 & bge-base & bge-m3 & MiniLM-L6 & bge-base & bge-m3 \\
\midrule
WVS & \cellcolor{shadeblue!28}0.0833 & \cellcolor{shadeblue!28}0.0417 & \cellcolor{shadeblue!4}0.0000 & \cellcolor{accentred!28}\dcell{+0.0833}{$+0.000$}{$+0.250$} & \cellcolor{accentred!20}\dcell{+0.0417}{$+0.000$}{$+0.125$} & \cellcolor{shadeblue!28}\dcell{$-0.0417$}{$-0.125$}{$+0.000$} \\
Perspectrum & \cellcolor{shadeblue!4}0.0691 & \cellcolor{shadeblue!28}0.0797 & \cellcolor{shadeblue!16}0.0730 & \cellcolor{accentred!8}\dcell{+0.0161}{$-0.001$}{$+0.035$} & \cellcolor{accentred!12}\dcell{\textbf{+0.0251}}{\textbf{+0.004}}{\textbf{+0.049}} & \cellcolor{accentred!16}\dcell{+0.0175}{$-0.010$}{$+0.048$} \\
ESS10 & \cellcolor{shadeblue!12}0.1573 & \cellcolor{shadeblue!28}0.2065 & \cellcolor{shadeblue!4}0.1241 & \cellcolor{accentred!28}\dcell{+0.0466}{$-0.022$}{$+0.118$} & \cellcolor{accentred!28}\dcell{\textbf{+0.0741}}{\textbf{+0.001}}{\textbf{+0.145}} & \cellcolor{accentred!20}\dcell{+0.0286}{$-0.028$}{$+0.085$} \\
\bottomrule
\end{tabular}
\end{table}

Table~\ref{tab:crossencoder} compares the canonical MiniLM-L6 pipeline (full-pool protocol) against two further bi-encoders (bge-base, bge-m3) under the earlier local-selection protocol. The gains depend on the embedding model. With bge-base, LMVU improves over MMR by $0.025$ RVC on Perspectrum and $0.074$ on ESS10, with bootstrap intervals excluding zero; bge-m3 yields smaller positive mean gains on both, with intervals crossing zero, and a negative mean difference on WVS ($-0.0417$). These evaluations rebuild embeddings, pools, and features together, characterizing pipeline sensitivity. Full scores and uncertainty estimates are in Appendix~\ref{app:embedder}.

\subsection{Sensitivity and operating boundaries}

\begin{figure}[t]
\centering
\captionsetup[subfigure]{font={small,bf},skip=2pt,position=top}
\begin{subfigure}[t]{0.62\linewidth}
\centering
\caption{RVC at each token budget}
\label{fig:budgets-sweep}
\includegraphics[width=\linewidth]{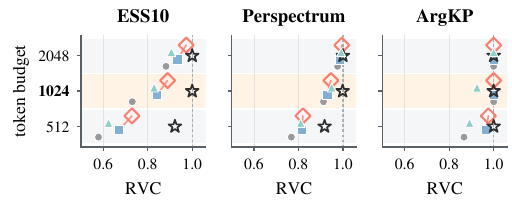}
\end{subfigure}\hfill%
\begin{subfigure}[t]{0.355\linewidth}
\centering
\caption{Headroom vs. realized gain}
\label{fig:budgets-gain}
\includegraphics[width=\linewidth]{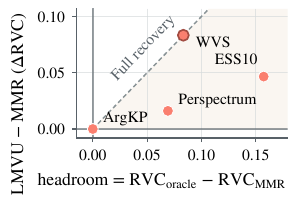}
\end{subfigure}
\par\vspace{2pt}
\includegraphics[width=\linewidth]{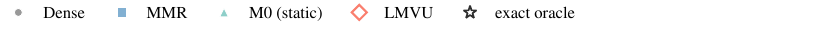}
\caption{(a)~Held-out RVC by budget (stars: exact oracle); ArgKP saturates at 1024/2048; LMVU reaches the WVS oracle at all budgets. (b)~Headroom--gain at 1024 (Pearson $r{=}0.574$, $n{=}4$).}
\label{fig:budgets}
\end{figure}

Figure~\ref{fig:budgets} shows held-out RVC by budget and the headroom--gain relation (full sweep inAppendix~\ref{app:dynamics}), the ESS10 gain persists at 2048; Perspectrum and saturated ArgKP show distinct patterns. Across benchmarks the realized gains track available headroom (Pearson $r{=}0.574$, $n{=}4$, descriptive) and vanish where the ceiling is met---the suite maps the method's operating region, not a universal win.

\paragraph{Budget scaling and gain concentration.}
The WVS gain over MMR is largest at 512 tokens ($+0.117$), persists at 1024 ($+0.083$), and narrows at 2048 ($+0.042$) as the relaxed budget lets MMR buy the missing viewpoints; ESS10 decays similarly ($+0.057/+0.047/+0.040$) but stays positive, while ArgKP's 512-token edge ($+0.010$) vanishes at saturation and Perspectrum's stays small (at most $+0.016$). At 1024 the WVS gain concentrates in the corruption-perception issue group ($+0.250$), whose evidence MMR's redundancy penalty crowds out; the remaining groups are solved by both allocators, leaving corruption as the sole margin.

%% file: sections/discussion_limitations.tex
\section{Discussion and Conclusion}
\label{sec:discussion}

LMVU learns state-dependent evidence selection under a token budget, training on small development pools and selecting from full query-specific pools. Across four benchmarks, it improves viewpoint coverage over MMR on three, reaches the exact coverage ceiling on WVS, and matches saturated ArgKP. State conditioning benefits WVS and ESS10, while a candidate-only scorer remains stronger on Perspectrum. Gains over MMR also extend to repeated evidence support.

\paragraph{Limitations.}
Performance depends on the embedding pipeline and viewpoint annotations for training and evaluation. The full-pool extension is exploratory (Appendix~\ref{app:protocol}), and the learned scorer has no coverage approximation guarantee. Further work should evaluate downstream generation and adaptation to domains requiring new viewpoint inventories.

%% file: sections/appendix.tex
\section{Evaluation Protocol and Integrity}
\label{app:protocol}

\paragraph{Development/confirmation discipline.}
All method decisions (feature schema, candidate pool, allocator, budgets, training configuration) were locked in development stages before any held-out access. The protected WVS fold (USA fold 4, $n{=}6$ queries) then saw a single-touch confirmation evaluation, followed only by the five predeclared baselines-only augmentations described below; no retuning occurred or was permitted at any point. Subsequently, the fixed-pool oracle was recomputed as an exact budgeted maximum-coverage optimum (\S\ref{sec:problem}); this required one additional read-only pass over the same gold annotations and changed no method, selection, tuning, or method-side number. The repeated-support audit of Appendix~\ref{app:rvc2} followed the same pattern: it recomputed per-trajectory coverage under $\RVCtwo$ and the exact repeated-support oracle from the locked selection records, re-accessing the protected WVS gold once more, read-only, to recover the augmentation baselines' selection sequences by deterministic replay. All twelve deterministic baselines reproduced their locked trajectories' RVC values bit-exactly; the random baseline was redrawn under a pinned seed with a predeclared sampling-noise gate (Appendix~\ref{app:rvc2}). No method, pool, tuning, or previously reported number changed. Each transfer benchmark followed: adapter verification $\rightarrow$ fixed candidate pool $\rightarrow$ development training (grouped by query) $\rightarrow$ development method lock $\rightarrow$ pre-gold trajectory lock $\rightarrow$ single-touch held-out evaluation.

\paragraph{Label-blindness firewall.}
Runtime selection uses only the 18 runtime-visible features. For every benchmark, the runtime access audits recorded zero gold-field accesses and zero external provider calls. Gold annotations were used exclusively to construct development-side training targets, to guide the gold-visitation state-harvesting family on development queries (\S\ref{sec:method}), and to compute the post-hoc coverage metrics used in the audit.

\paragraph{Determinism and seeds.}
The allocator is deterministic; across ten selection seeds $\{17, 29, 41, 53, 67, 79, 83, 97, 109, 127\}$ the action sequences are identical on all checked queries. Training uses a fixed seed (41). All similarity geometry uses the frozen sentence-transformers/all-MiniLM-L6-v2 encoder (snapshot \texttt{1110a243}, 384 dimensions, maximum sequence 256 tokens, normalized embeddings, local files only); the pipeline fails closed on any deviation from this identity. The full-pool extension fits three training seeds $\{17, 29, 41\}$ for the learned rows (M0, JPR, RichRAG, LMVU); given a fitted scorer, inference is deterministic and runs once per query, and the stochastic baselines (Random, xQuAD, Cluster) use the ten pinned seeds above.

\paragraph{Stopping rule.}
Dataset expansion was governed by a predeclared rule: stop when two WVS-external benchmarks reach positive status (met: Perspectrum, ESS10) or after at most three justified headroom-rich additions. The rule triggered on the first condition.

\paragraph{Classification tiers.}
Held-out outcomes use predeclared ordered tiers (strong positive / positive / boundary non-inferior / negative) defined by delta sign, magnitude, and headroom recovery; saturated non-inferiority does not count toward the positive target.

\paragraph{Baseline augmentation locks.}
The thirteen added baselines were predeclared in five baselines-only augmentation locks before their evaluation: \textsc{Random}, \textsc{BM25}, \textsc{FacLoc} (lock v1); \textsc{xQuAD}, \textsc{kDPP}, \textsc{SubMix}, \textsc{LexRank} (lock v2); \textsc{CeRerank}, \textsc{Cluster} (lock v3); \textsc{AdaGReS} \citep{peng2025adagres}, \textsc{MMKP} \citep{xu2026selfcorrecting} (lock v4); and \textsc{JPR} \citep{min2021jpr}, \textsc{RichRAG} \citep{wang2025richrag} (lock v5). Lock v4 ran the published definitions on the same frozen MiniLM pool geometry as every other baseline: \textsc{AdaGReS} greedily adds the item maximizing query relevance minus an instance-adaptive redundancy penalty estimated from pool-level similarity means, stopping at a non-positive marginal gain; \textsc{MMKP} clusters each pool at a $0.82$ cosine threshold (connected components), admits at most one item per cluster, and solves the resulting budgeted relevance maximization exactly by Pareto-pruned dynamic programming; both are deterministic. Lock v5 adapts the two published neural rankers to state-aware selection on the same frozen pools and development states: \textsc{JPR} is a t5-base sequence-to-sequence reranker trained with a listwise softmax cross-entropy over the \emph{yes} logit on inputs of the form query $+$ current selection $+$ candidate, and \textsc{RichRAG} is a Qwen2.5-7B-Instruct LoRA adapter trained by supervised fine-tuning on yes/no new-perspective judgments followed by direct preference optimization on same-state winner/loser candidate pairs; at inference each greedily adds its highest-scoring candidate, scoring only the top-64 dense-relevant budget-feasible candidates per step, with per-benchmark fits (WVS per dependency fold) under the three training seeds and deterministic forward-pass inference. For \textsc{CeRerank}, the cross-encoder (bge-reranker-v2-m3, revision 953dc6f) scored every frozen pool item exactly once before selection; the scores were cached with their input hashes and shared across all seeds, making the baseline deterministic, and the model is self-hosted with zero external provider calls. Each lock fixed the baseline definitions, budgets $\{512,1024,2048\}$, and seeds $\{17,\dots,127\}$ before any run, with no change to the frozen LMVU method, pools, or previously reported results, and no retuning. The WVS protected fold was re-accessed for these fixed baselines only: once per lock for v1--v3, for lock v4 once under the original local-pool protocol and once under this version's full-pool protocol, both fixed by the same lock, and for lock v5 once under the full-pool protocol. LMVU's protected result was neither recomputed nor modified.

\paragraph{Full-pool extension (this version).}
This version changes where selection happens: training remains confined to the frozen local pools $P(q)$, while all reported evaluation runs select from the full prequalified evidence pools $E(q)$. The model class, objective, and training configuration are unchanged, but training states are re-harvested under the delivered pipeline (\S\ref{sec:method})---budget-indexed local pools and five development state-visitation families---and every learned row (WVS per dependency fold, each transfer benchmark, M0 under its three-feature mask) is refit on these states with three training seeds $\{17, 29, 41\}$. The extension is built on an independently delivered LF (local-training/full-selection) matrix whose bundle passed integrity verification (\texttt{verify\_bundle}), whose $5{,}148$ learned (LMVU) trajectories were bit-reproduced by parameter replay (of $7{,}293$ delivered runs), and whose MMR and Dense rows were reconciled per query (tolerance $10^{-9}$) against an independent offline rerun on all four benchmarks; that rerun reproduced all deterministic rows exactly and flipped only three near-tie LMVU rows, traced to cross-version refit noise, and none of this version's numbers is drawn from those rerun rows. The sixteen non-oracle baselines were then rerun on the full pools under their locked scorer definitions---the original twelve plus the lock-v4 \textsc{AdaGReS} and \textsc{MMKP} and the lock-v5 \textsc{JPR} and \textsc{RichRAG}, ported onto the same pool matrix---gated on per-query reproduction ($\le 10^{-9}$) of the delivered MMR and Dense rows; \textsc{CeRerank} reused the same self-hosted cross-encoder, and the stochastic baselines were redrawn under ten pinned seeds. The exact oracles (breadth and repeated support) were recomputed on the full pools; the repeated-support oracle was verified by an in-process identical second run, dominates every method trajectory on every query--budget cell, and reproduces the locked local-pool audit exactly when restricted to the original pools. Because this extension follows previously inspected held-out results, its evaluations are exploratory rather than a fresh confirmatory test. All development-fold diagnostics (Appendices~\ref{app:dev}, \ref{app:learn}, \ref{app:sensitivity}, \ref{app:ablations}, \ref{app:features}), the embedding sweep (Table~\ref{tab:embedder}) and the bge encoder columns of Table~\ref{tab:crossencoder}, and the $\Delta_2$ ablation of Appendix~\ref{app:rvc2} are artifacts of the original local-selection development pipeline and are reported unchanged. The case-study tables and figures in Appendix~\ref{app:cases} instead display the evaluation trajectories obtained by selecting from the full pools in this version.

\section{Experimental Setup}
\label{app:setup}

\paragraph{Datasets.}
Table~\ref{tab:datasets} summarizes the four viewpoint-annotated benchmarks. WVS \citep{zhao2024worldvaluesbench} contributes survey topics with component-level atomic claims; we evaluate on a protected respondent fold under the access discipline of \S\ref{app:protocol}. Perspectrum \citep{chen2019seeing} contributes claim--perspective units, ArgKP \citep{barhaim2020arguments} argument key points (included as a saturation control), and ESS10 \citep{ess2020round10} European Social Survey round-10 items. Development and held-out splits are grouped by query so that no query appears in both.

\begin{table}[h]
\caption{Benchmark summary. Valid viewpoints/query and full selection-pool rows/query are means (the full pools are budget-independent); training uses the smaller MMR-oversampled pools below. WVS dev and evaluation use the same six queries (3 topics $\times$ 2 paraphrases) but disjoint gold matrices; evaluation uses the protected respondent fold. Cell shading scales by value per column.}
\label{tab:datasets}
\begin{center}
\small
\begin{tabular}{lcccc}
\toprule
Benchmark & Dev $n$ & Eval $n$ & Valid viewpoints/query & Pool rows/query \\
\midrule
WVS (protected) & \cellcolor{shadeblue!4}6 & \cellcolor{shadeblue!4}6 & \cellcolor{shadeblue!28}6.0 & \cellcolor{shadeblue!28}507.3 \\
Perspectrum & \cellcolor{shadeblue!8}32 & \cellcolor{shadeblue!28}97 & \cellcolor{shadeblue!4}3.0 & \cellcolor{shadeblue!4}8.9 \\
ArgKP & \cellcolor{shadeblue!8}32 & \cellcolor{shadeblue!4}6 & \cellcolor{shadeblue!24}5.5 & \cellcolor{shadeblue!8}75.7 \\
ESS10 & \cellcolor{shadeblue!28}152 & \cellcolor{shadeblue!12}34 & \cellcolor{shadeblue!20}4.8 & \cellcolor{shadeblue!8}88.3 \\
\bottomrule
\end{tabular}
\end{center}
\end{table}

\paragraph{Pool construction (training and selection).}
All methods share identical pools. The \emph{selection} pool $E(q)$ is each query's full prequalified evidence set: deduplicated, budget-independent, with token costs from the adapter tokenizer's whole-evidence token lengths; the runtime candidate policy enumerates exactly the budget-feasible pool members. The smaller \emph{training} pool $P(q)$ is the frozen MMR-oversampled candidate set of the original pipeline (maximal marginal relevance with $\lambda{=}0.5$ over dense query--candidate scores and embedding novelty, filled to $4\times$ the budget in tokens) and is used only to harvest development training states (\S\ref{sec:method}). The pools are deliberately redundant (multiple rows per viewpoint; Table~\ref{tab:datasets}). Coverage therefore depends on which evidence the allocator selects, rather than being fixed by a minimally populated pool.

\paragraph{Training.}
Training targets are realized marginal coverage gains $\Delta V(e \mid q, S)$ on development states only. For each state we form within-state candidate pairs $(e^+, e^-)$ with $\Delta V(e^+) > \Delta V(e^-)$, capped at 20 pairs per state (deterministic hash-seeded subsample), and fit a logistic pairwise ranker on standardized symmetric feature differences ($C{=}1.0$, LBFGS, max.\ 300 iterations, seed 41). WVS fits one model per query-dependency fold and each transfer benchmark fits per-benchmark coefficients under the same configuration; every learned row reported in this version is refit on the harvested development states of \S\ref{sec:method} with three training seeds $\{17, 29, 41\}$. The hyperparameter setting is shared and predeclared across benchmarks, and no fitting uses held-out data.

\section{Baseline Definitions}
\label{app:baselines}

All baselines share the full selection pools and token accounting of this version's full-pool rerun (Appendix~\ref{app:protocol}); they differ only in the per-step scoring rule. \textsc{Random}: uniform random priority order per query (10 seeds, mean reported). \textsc{BM25}: hand-rolled BM25 ($k_1{=}1.2$, $b{=}0.75$) over the query's pool documents, score static across steps. \textsc{FacLoc}: facility-location greedy over frozen sentence embeddings; the per-step score of candidate $e$ is $\sum_{j \in E(q)} \max\big(0,\ \mathrm{sim}(e,j) - c_j\big)$ with $c_j$ the running best coverage of pool item $j$, a relevance-free monotone submodular coverage objective.

Four citation-backed baselines (lock v2) are local approximations over the same full pools. \textsc{xQuAD} \citep{santos2010exploiting}: explicit aspect diversification with pseudo-aspects from spherical $k$-means over pool embeddings ($k = \min(10, \max(2, \lceil\sqrt{|E|}\rceil))$, label-blind), $\lambda{=}0.5$, hard aspect membership. \textsc{kDPP} \citep{kulesza2012determinantal}: greedy MAP inference over an $L$-ensemble with quality $q_i$ the min-max-normalized dense relevance (shifted positive) and kernel $S$ the clipped cosine similarity; incremental Cholesky updates. \textsc{SubMix} \citep{lin2011class}: $\alpha{=}0.5$ mixture of normalized dense relevance and the normalized facility-location gain. One disclosure: with the gain divided by pool size, the coverage term's spread across candidates is $\sim\!10^{-3}$ against $\sim\!10^{-1}$ for the min-max relevance term on these pools, so the mixture degenerates to dense top-$k$ on WVS and on all but one ESS10 query at the primary budget (per-query RVC identical to full precision); we report the baseline exactly as predeclared in the lock rather than renormalizing post hoc. \textsc{LexRank} \citep{erkan2004lexrank}: PageRank centrality (damping $0.85$, 50 iterations) over the cosine-similarity graph of the pool thresholded at $0.1$; budget-feasible top-$k$.

Two modern baselines (lock v3) follow the same full-pool discipline. \textsc{CeRerank}: every pool item is scored once by the cross-encoder reranker bge-reranker-v2-m3 \citep{chen2024bge} (query--evidence text pairs truncated at 512 tokens, fp16 inference, scores cached with input hashes and shared across all ten seeds), and selection is budget-feasible greedy on the cached score. \textsc{Cluster}: spherical $k$-means over the frozen pool embeddings ($k = \min(12, \max(2, \lceil\sqrt{|E|}\rceil))$, seeded initialization, 25 iterations); clusters are ordered by the dense relevance of their most relevant member, and selection rotates round-robin across clusters, taking each cluster's next most relevant item, yielding a static rank order that is consumed budget-feasibly.

Two allocation heuristics (lock v4) and two training-heavy neural rankers (lock v5) complete the panel; Appendix~\ref{app:protocol} gives their locked configurations. \textsc{AdaGReS} \citep{peng2025adagres}: greedy query relevance minus an instance-adaptive pairwise-similarity redundancy penalty, stopping at a non-positive marginal gain. \textsc{MMKP} \citep{xu2026selfcorrecting}: at most one item per $0.82$-cosine connected-component cluster, with exact Pareto-pruned dynamic programming for the budgeted relevance maximization. \textsc{JPR} \citep{min2021jpr}: t5-base sequence-to-sequence reranker trained with a listwise softmax over the \emph{yes} logit on query $+$ current selection $+$ candidate inputs. \textsc{RichRAG} \citep{wang2025richrag}: Qwen2.5-7B LoRA adapter trained by SFT on yes/no new-perspective judgments followed by DPO on same-state candidate pairs; both neural rankers score only the top-64 dense-relevant budget-feasible candidates per step.

\section{Paired Statistics Protocol}
\label{app:stats}

For each benchmark at budget 1024 we pair per-query RVC (seed-averaged where applicable: ten pinned seeds for the stochastic baselines, three training seeds for the learned rows) between LMVU and each comparator, then compute the mean delta, a paired bootstrap 95\% CI (10{,}000 replicates, fixed seed), win/tie/loss counts at tolerance $10^{-9}$, and the two-sided sign-test $p$-value. On WVS and ArgKP, whose six queries form only three independent issue/topic groups, no intervals are reported and statistics are descriptive; on Perspectrum and ESS10 a Wilcoxon signed-rank check accompanies the sign test. All quantities are computed from locked per-query metric artifacts without re-evaluation. Table~\ref{tab:rvc2_budgets} reports the paired statistics for LMVU $-$ MMR as the reference comparison (breadth at the primary budget; depth across budgets). Most queries tie near the coverage ceiling, so we emphasize effect sizes and uncertainty intervals rather than drawing separate significance claims for each benchmark.

\input{sections/analysis}

\section{Development-Fold Evidence}
\label{app:dev}

Development-fold cross-validation (grouped by query; never used for held-out claims) at budget 1024 (cell shading scales with RVC within each column):

\begin{center}
\small
\begin{tabular}{lcccc}
\toprule
Benchmark (dev $n$) & Dense & MMR & M0 & M1 (LMVU) \\
\midrule
Perspectrum (32) & \cellcolor{shadeblue!24}0.9182 & \cellcolor{shadeblue!16}0.8964 & \cellcolor{shadeblue!28}\textbf{0.9714} & \cellcolor{shadeblue!24}0.9583 \\
ArgKP (32) & \cellcolor{shadeblue!28}0.9771 & \cellcolor{shadeblue!28}0.9885 & \cellcolor{shadeblue!28}0.9797 & \cellcolor{shadeblue!28}\textbf{0.9896} \\
ESS10 (152) & \cellcolor{shadeblue!4}0.7318 & \cellcolor{shadeblue!4}0.8079 & \cellcolor{shadeblue!4}0.8381 & \cellcolor{shadeblue!4}\textbf{0.8408} \\
\bottomrule
\end{tabular}
\end{center}

The development pattern matches the held-out one: M1 improves over MMR on Perspectrum ($+0.062$) and ESS10 ($+0.033$), the M0-over-M1 inversion on Perspectrum is already visible in development, and gains are small on near-saturated ArgKP.

\section{Pre-Registered Learnability Probe}
\label{app:learn}

Before any sequential evaluation we tested, on development states only, whether marginal viewpoint utility is learnable at all. The probe paired 117{,}387 runtime-feature rows (7{,}122 states, 70 development queries across WVS, Perspectrum, and ArgKP) with hidden marginal-gain targets, and compared frozen heuristic rankers and the learned models under grouped cross-validation (identical states and candidate sets for every method). Table~\ref{tab:learnability} reports the within-state ranking quality measured by the probe. The pre-registered gate required M1 to outperform M0 on NDCG@5, state oracle regret, and top-1 positive-gain rate; the probe met all three criteria.

\begin{table}[h]
\caption{Learnability probe on development states pooled over WVS, Perspectrum, and ArgKP: within-state NDCG and mean state oracle regret (lower is better). The heuristic rankers are frozen, while M0 and M1 are learned. The oracle ranks candidates by their true marginal gain. Shading is column-wise over the seven non-oracle rows, with regret reversed so that darker means lower regret. The oracle row has a uniform gray background and is excluded from each column's color scale.}
\label{tab:learnability}
\begin{center}
\small
\begin{tabular}{lccc}
\toprule
Method & NDCG@5 & NDCG@10 & Oracle regret \\
\midrule
Random & \cellcolor{shadeblue!12}0.5554 & \cellcolor{shadeblue!12}0.6165 & \cellcolor{shadeblue!8}0.5638 \\
Dense relevance & \cellcolor{shadeblue!4}0.5243 & \cellcolor{shadeblue!4}0.5837 & \cellcolor{shadeblue!4}0.5804 \\
MMR-style & \cellcolor{shadeblue!24}0.6313 & \cellcolor{shadeblue!24}0.6925 & \cellcolor{shadeblue!24}0.4789 \\
Soft novelty & \cellcolor{shadeblue!24}0.6223 & \cellcolor{shadeblue!24}0.6801 & \cellcolor{shadeblue!28}0.4706 \\
Static validity & \cellcolor{shadeblue!12}0.5707 & \cellcolor{shadeblue!12}0.6308 & \cellcolor{shadeblue!12}0.5393 \\
M0 (static learned) & \cellcolor{shadeblue!16}0.5869 & \cellcolor{shadeblue!16}0.6513 & \cellcolor{shadeblue!8}0.5525 \\
M1 (LMVU) & \cellcolor{shadeblue!28}\textbf{0.6575} & \cellcolor{shadeblue!28}\textbf{0.7165} & \cellcolor{shadeblue!28}\textbf{0.4516} \\
\midrule
Oracle marginal gain & \cellcolor{oraclegray}1.0000 & \cellcolor{oraclegray}1.0000 & \cellcolor{oraclegray}0.0000 \\
\bottomrule
\end{tabular}
\end{center}
\end{table}

Two observations motivate the method design. First, no runtime-visible heuristic approaches the oracle, and dense relevance is \emph{worse than random} on within-state ranking---relevance does not order marginal gains. Second, the state-dependent model (M1) is the only method that beats every heuristic on all three criteria, which is what licensed sequential training.

\section{Feature Ablations}
\label{app:ablations}

\paragraph{Static vs.\ state-dependent.}
M0 (three static features) vs.\ M1 (full state-dependent features) within-state ranking on development states (cell shading scales with value within each column; regret columns flipped, darker is lower):

\begin{center}
\small
\begin{tabular}{lcccc}
\toprule
& \multicolumn{2}{c}{NDCG@5} & \multicolumn{2}{c}{Oracle regret} \\
\cmidrule(lr){2-3}\cmidrule(lr){4-5}
Dataset & M0 & M1 & M0 & M1 \\
\midrule
WVS & \cellcolor{shadeblue!4}0.3961 & \cellcolor{shadeblue!16}\textbf{0.6597} & \cellcolor{shadeblue!4}1.3945 & \cellcolor{shadeblue!4}\textbf{0.8554} \\
Perspectrum & \cellcolor{shadeblue!28}0.8393 & \cellcolor{shadeblue!28}\textbf{0.8413} & \cellcolor{shadeblue!28}\textbf{0.2537} & \cellcolor{shadeblue!28}0.2688 \\
ArgKP & \cellcolor{shadeblue!4}0.3704 & \cellcolor{shadeblue!4}\textbf{0.4732} & \cellcolor{shadeblue!20}0.6934 & \cellcolor{shadeblue!16}\textbf{0.5587} \\
\bottomrule
\end{tabular}
\end{center}

State-dependence drives the ranking-quality gain on WVS and ArgKP; on Perspectrum the two are near-identical, foreshadowing the held-out M0$>$M1 inversion discussed in \S6.

\paragraph{Atomic features.}
Removing the six atomic-granularity features from M1 (12-feature variant, retrained on the same development states) leaves Perspectrum and ArgKP unchanged (NDCG@5 deltas $\le 0.0002$) but sharply degrades WVS: NDCG@5 $0.6597 \rightarrow 0.4418$, top-1 positive-gain rate $0.5425 \rightarrow 0.4088$. This WVS result highlights the value of retaining atomic structure when individual survey statements express more than one distinct claim.

\section{Parameter Sensitivity (Development Only)}
\label{app:sensitivity}

Figure~\ref{fig:sensitivity} shows dev-CV RVC at budget 1024 when varying the two training-side hyperparameters around the frozen operating point (20 pairs/state, $C{=}1.0$; reproduced exactly: Perspectrum 0.9583, ESS10 0.8408). ESS10 is flat across both axes (RVC within $[0.838, 0.842]$); Perspectrum varies mildly with the training-pair cap ($0.9406$ at 5 pairs, $0.9625$ at 10, $0.9583$ at 20) and is insensitive to $C$ over two orders of magnitude ($0.9583$ at $C \in \{0.1, 1, 10\}$). The method does not rely on a finely tuned operating point. These runs accessed development folds only; no held-out or protected data were accessed during the sensitivity analysis.

\begin{figure}[t]
\begin{center}
\includegraphics[width=0.95\linewidth]{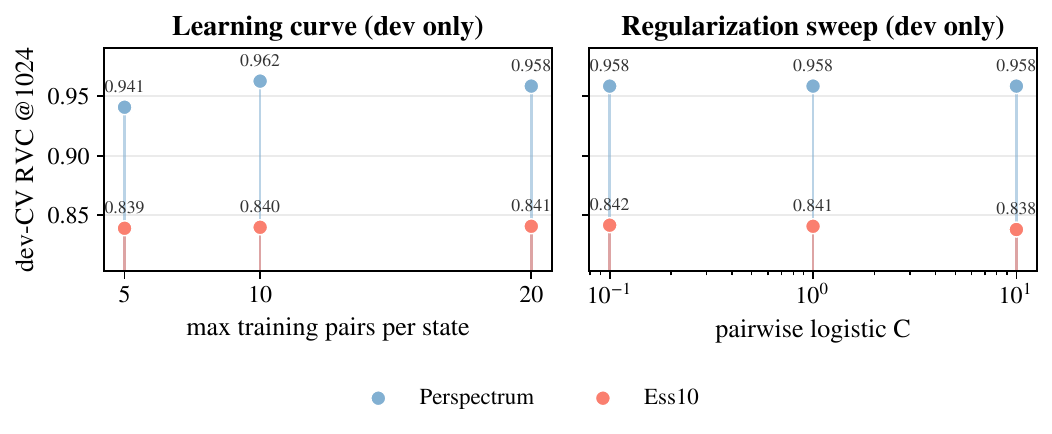}
\end{center}
\caption{Training-side sensitivity on development folds (grouped CV, budget 1024). Left: learning curve over the per-state pair cap. Right: regularization sweep.}
\label{fig:sensitivity}
\end{figure}

\section{Embedding-Model Robustness}
\label{app:embedder}

The canonical pipeline freezes one encoder (sentence-transformers/all-MiniLM-L6-v2; \S\ref{sec:experiments}). We predeclared a sweep over bge-base-en-v1.5 \citep{xiao2023cpack} and bge-m3 \citep{chen2024bge}, rerunning embeddings, pools, baselines, training, and the exact oracle at budgets $\{512,1024,2048\}$ with the locked configuration. A reimplementation gate reproduced canonical pools and MMR/oracle macros to $10^{-9}$. WVS's protected-fold row used a predeclared evaluation-only amendment with no fitting or tuning. Table~\ref{tab:crossencoder} summarizes the primary-budget comparison; Table~\ref{tab:embedder} gives full results.

\begin{table}[t]
\caption{Embedding-model robustness at budget 1024 in the main-text scopes (WVS: protected fold; Perspectrum, ArgKP, ESS10: held out): exact-oracle headroom (oracle $-$ MMR), RVC of M1 and MMR, and the paired M1 $-$ MMR delta with 95\% bootstrap CI and win/tie/loss counts. All rows are retained from the original local-selection protocol (smaller candidate pools; Appendix~\ref{app:protocol}). A panel-level summary appears as Table~\ref{tab:crossencoder}; WVS development-fold metrics remain in the sweep artifact. Cell shading scales with value within each column (headroom by magnitude; M1/MMR darker is higher; the $\Delta$ column diverging: red positive, blue negative, zero unshaded).}
\label{tab:embedder}
\begin{center}
\scriptsize
\setlength{\tabcolsep}{2.5pt}
\begin{tabular}{llccccc}
\toprule
Benchmark & Encoder & Headroom & M1 & MMR & $\Delta$ [95\% CI] & W/T/L \\
\midrule
WVS & MiniLM-L6 (canonical) & \cellcolor{shadeblue!8}0.0417 & \cellcolor{shadeblue!24}0.9583 & \cellcolor{shadeblue!20}0.9167 & \cellcolor{accentred!20}$+0.0417$ [$+0.0000$, $+0.1250$] & 1/5/0 \\
 & bge-base-en-v1.5 & \cellcolor{shadeblue!8}0.0417 & \cellcolor{shadeblue!16}0.9167 & \cellcolor{shadeblue!16}0.8750 & \cellcolor{accentred!20}$+0.0417$ [$+0.0000$, $+0.1250$] & 1/5/0 \\
 & bge-m3 & \cellcolor{shadeblue!4}0.0000 & \cellcolor{shadeblue!16}0.9167 & \cellcolor{shadeblue!24}0.9583 & \cellcolor{shadeblue!20}$-0.0417$ [$-0.1250$, $+0.0000$] & 0/5/1 \\
\midrule
Perspectrum & MiniLM-L6 (canonical) & \cellcolor{shadeblue!12}0.0663 & \cellcolor{shadeblue!20}0.9392 & \cellcolor{shadeblue!20}0.9296 & \cellcolor{accentred!8}$+0.0096$ [$-0.0069$, $+0.0283$] & 7/85/5 \\
 & bge-base-en-v1.5 & \cellcolor{shadeblue!16}0.0797 & \cellcolor{shadeblue!20}0.9424 & \cellcolor{shadeblue!20}0.9173 & \cellcolor{accentred!12}$+0.0251$ [$+0.0042$, $+0.0491$] & 10/84/3 \\
 & bge-m3 & \cellcolor{shadeblue!12}0.0730 & \cellcolor{shadeblue!20}0.9420 & \cellcolor{shadeblue!20}0.9246 & \cellcolor{accentred!12}$+0.0175$ [$-0.0100$, $+0.0475$] & 12/78/7 \\
\midrule
ArgKP & MiniLM-L6 (canonical) & \cellcolor{shadeblue!4}0.0000 & \cellcolor{shadeblue!28}1.0000 & \cellcolor{shadeblue!28}1.0000 & $+0.0000$ [$+0.0000$, $+0.0000$] & 0/6/0 \\
 & bge-base-en-v1.5 & \cellcolor{shadeblue!4}0.0000 & \cellcolor{shadeblue!28}1.0000 & \cellcolor{shadeblue!28}1.0000 & $+0.0000$ [$+0.0000$, $+0.0000$] & 0/6/0 \\
 & bge-m3 & \cellcolor{shadeblue!4}0.0000 & \cellcolor{shadeblue!28}1.0000 & \cellcolor{shadeblue!28}1.0000 & $+0.0000$ [$+0.0000$, $+0.0000$] & 0/6/0 \\
\midrule
ESS10 & MiniLM-L6 (canonical) & \cellcolor{shadeblue!24}0.1465 & \cellcolor{shadeblue!8}0.8866 & \cellcolor{shadeblue!12}0.8427 & \cellcolor{accentred!20}$+0.0438$ [$-0.0220$, $+0.1142$] & 9/19/6 \\
 & bge-base-en-v1.5 & \cellcolor{shadeblue!28}0.2065 & \cellcolor{shadeblue!4}0.8617 & \cellcolor{shadeblue!4}0.7876 & \cellcolor{accentred!28}$+0.0741$ [$+0.0010$, $+0.1452$] & 14/15/5 \\
 & bge-m3 & \cellcolor{shadeblue!20}0.1241 & \cellcolor{shadeblue!12}0.9046 & \cellcolor{shadeblue!16}0.8759 & \cellcolor{accentred!16}$+0.0286$ [$-0.0277$, $+0.0853$] & 9/19/6 \\
\bottomrule
\end{tabular}
\end{center}
\end{table}

Two readings, both conditional on the representation. First, headroom remains necessary in every cell: across the 24 encoder$\times$benchmark$\times$budget cells of the two reported encoders, the seven with zero headroom (all ArgKP cells except bge-base at 512; bge-base WVS at 2048; bge-m3 WVS at 1024) never admit a positive delta---where MMR is already oracle-optimal there is nothing to recover. Second, headroom is not sufficient: of the 17 positive-headroom cells the delta is positive in 14 and negative in three, and all three negative cells arise under bge-m3 (Perspectrum at 512 and 2048; ESS10 at 2048) with intervals that include zero. On the held-out positive-headroom cells the delta is positive in 10 of 13, with four intervals excluding zero on the positive side and none on the negative side---all four under bge-base-en-v1.5 (Perspectrum at 512 and 1024; ESS10 at 512 and 1024), which also recovers the small ArgKP headroom at 512 in full and reproduces the original-protocol $+0.042$ on the WVS protected fold at the primary budget. Under bge-m3 the pattern is directionally preserved but attenuated: positive deltas at the primary budget on Perspectrum and ESS10 with intervals including zero, and one negative protected-fold cell (WVS at 1024, $-0.042$, interval including zero; its 512 and 2048 cells are $+0.042$). The original-protocol result replicates under the encoder closest in family and scale but attenuates under the more distant representation; the frozen encoder is part of the method's operating point, and we claim no representation-free gains.

\section{Failure-Case Tally}
\label{app:failures}

From the development-state diagnostic (same states as Appendix~\ref{app:learn}), we tally states where M1's chosen action underperforms. \emph{M1 beats MMR}: 195 (WVS), 67 (Perspectrum), 432 (ArgKP) states. \emph{M1 underperforms MMR}: 34, 65, 595. \emph{M1 takes a zero-gain action while a positive-gain action exists}: 200, 336, 2113---concentrated in mid/late allocation steps (pooled: 528 early / 1157 mid / 964 late), where most remaining candidates add no new viewpoint. Failure mass is thus largest on ArgKP, the benchmark where everything saturates and the ordering error is inconsequential, and smallest relative to opportunity on WVS, where the held-out gain is largest.

\input{sections/cases}

\section{Feature List}
\label{app:features}

The 18 runtime-visible, label-blind features of M1 fall into four groups: candidate descriptors (4), evidence--state geometry (5), atomic--state geometry (5), and shared state descriptors (4). Table~\ref{tab:features_full} gives the full definitions. All geometry is computed from incremental sufficient statistics maintained over the selection---running similarity extrema, the selected-embedding sum, and per-atom maxima---so no pairwise pool-similarity matrix is materialized; an empty selection yields similarity $0$ and novelty $1$ by convention. M0 uses three candidate descriptors: dense relevance, query--candidate cosine, and candidate token length. Figure~\ref{fig:coef} profiles the fitted coefficients of the deployed ESS10 models (local training pools, three training seeds), and two readings matter. First, the largest-magnitude weights sit on the evidence--state geometry group, while the static candidate descriptors carry small weights---consistent with the M0 ablation. Second, the four shared state descriptors are constant across candidates within a state: they cancel in the within-state pairwise differences of \S\ref{sec:method:ranking} and cannot change the argmax of \S\ref{sec:method:inference}, and their fitted coefficients are exactly zero in every fitted model of this version. They are recorded for completeness; the budget binds through the feasibility filter of \S\ref{sec:method:inference}, not through the learned score.

\begin{table}[h]
\caption{The 18 runtime-visible, label-blind features $\phi(q, S, e)$ (max/mean/min variants counted separately), grouped as in Figure~\ref{fig:coef}. $S$ denotes the current selection and $u$ the tokens it uses; $z_x$ is the frozen embedding of unit $x$; $\cos$ is cosine over the frozen embeddings; $\mathrm{rel}$ is the dense retriever score; $\bar{z}_S = \frac{1}{|S|}\sum_{e' \in S} z_{e'}$ is the unnormalized mean selected vector; $a \in A(e)$ are the benchmark's atomic sub-evidence units of $e$, $A(S) = \bigcup_{e' \in S} A(e')$ the selected ones, and $m_a = \max_{a' \in A(S)} \cos(a, a')$.}
\label{tab:features_full}
\begin{center}
\small
\setlength{\tabcolsep}{4pt}
\begin{tabular}{@{}ll@{}}
\toprule
\multicolumn{2}{@{}l}{\textit{candidate descriptors}}\\
\midrule
dense relevance & $\mathrm{rel}(q, e)$ \\
query--candidate cosine & $\cos(q, e)$ \\
candidate token length & $c(e)$ \\
candidate atomic count & $|A(e)|$ \\
\midrule
\multicolumn{2}{@{}l}{\textit{evidence--state geometry}}\\
\midrule
sim.\ to selected evidence & $\max / \operatorname{mean} / \min_{e' \in S} \cos(e, e')$ \\
novelty vs.\ evidence & $1 - \max_{e' \in S} \cos(e, e')$ \\
candidate centroid sim. & $z_e^{\top} \bar{z}_S$ \\
\midrule
\multicolumn{2}{@{}l}{\textit{atomic--state geometry}}\\
\midrule
sim.\ to selected atomic units & $\max / \operatorname{mean}_{a \in A(e)} m_a$ \\
novelty vs.\ atomic units & $1 - \max_{a \in A(e)} m_a$ \\
atomic novelty & $\operatorname{mean} / \max_{a \in A(e)} \bigl[1 - m_a\bigr]$ \\
\midrule
\multicolumn{2}{@{}l}{\textit{shared state descriptors}}\\
\midrule
selected evidence count & $|S|$ \\
tokens used / remaining & $u$,\; $B - u$ \\
mean selected relevance & $\operatorname{mean}_{e' \in S} \mathrm{rel}(q, e')$ \\
\bottomrule
\end{tabular}
\end{center}
\end{table}

\begin{figure}[h]
\begin{center}
\includegraphics[width=\linewidth]{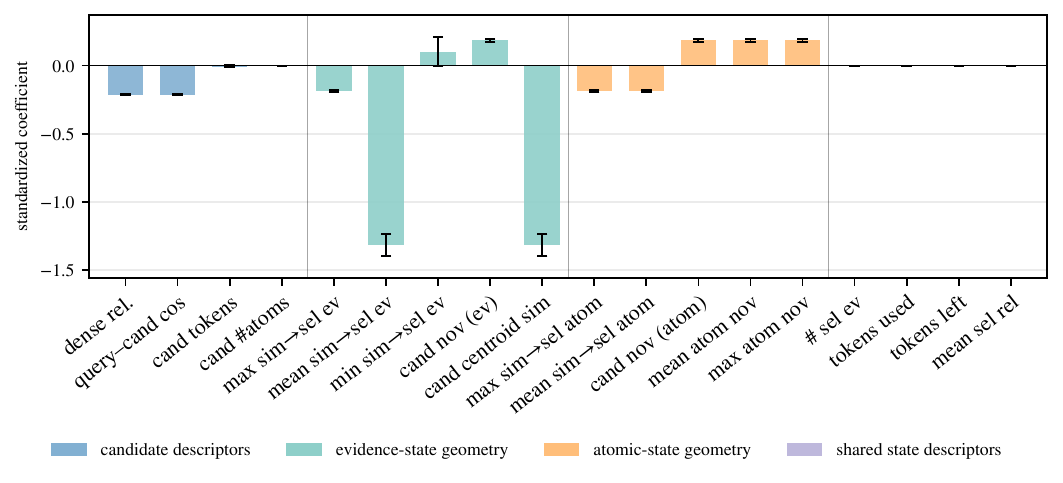}
\end{center}
\caption{M1 coefficient profile over the 18 runtime-visible features, grouped as in \S\ref{sec:method}; mean $\pm$ sd across the three training seeds of the deployed ESS10 model (local training pools; feature labels are abbreviated and presented in the order listed in Appendix~\ref{app:features}).}
\label{fig:coef}
\end{figure}

\section{Per-Benchmark Structure}
\label{app:structure}

Full selection-pool rows per query (min/mean/max): WVS $506/507.3/508$, Perspectrum $4/8.9/28$, ArgKP $47/75.7/91$, ESS10 $34/88.3/119$. MMR token utilization at budget 1024 is $0.85$--$0.97$ per benchmark (Table~\ref{tab:structure})---budgets are genuinely binding. Positive-marginal-action frequency under M1: 0.321 (WVS), 0.520 (Perspectrum), 0.110 (ArgKP), 0.404 (ESS10).

\begin{table}[t]
\caption{Structural moderators per benchmark (descriptive; $n=4$, no causal claim). Headroom $= \RVC_{\text{oracle}} - \RVC_{\text{MMR}}$ at budget 1024 on the full selection pools, and Headroom$_2$ the same difference under $\RVCtwo$; rows/viewpoint measures full-pool redundancy; MMR util.\ is MMR's mean token utilization at budget 1024. Cell shading scales with value within each column (LMVU$\,-\,$MMR shaded red by magnitude, zero unshaded).}
\label{tab:structure}
\begin{center}
\scriptsize
\setlength{\tabcolsep}{2.4pt}
\begin{tabular}{lcccccc}
\toprule
Benchmark & Headroom & Headroom$_2$ & LMVU$\,-\,$MMR & Rows/viewpoint & Viewpoints/query & MMR util. \\
\midrule
WVS (protected) & \cellcolor{shadeblue!16}0.083 & \cellcolor{shadeblue!28}0.218 & \cellcolor{accentred!28}+0.083 & \cellcolor{shadeblue!28}84.6 & \cellcolor{shadeblue!28}6.0 & \cellcolor{shadeblue!20}0.931 \\
Perspectrum & \cellcolor{shadeblue!16}0.069 & \cellcolor{shadeblue!12}0.074 & \cellcolor{accentred!8}+0.016 & \cellcolor{shadeblue!4}3.0 & \cellcolor{shadeblue!4}3.0 & \cellcolor{shadeblue!4}0.853 \\
ArgKP & \cellcolor{shadeblue!4}0.000 & \cellcolor{shadeblue!4}0.000 & 0.000 & \cellcolor{shadeblue!8}13.8 & \cellcolor{shadeblue!24}5.5 & \cellcolor{shadeblue!28}0.968 \\
ESS10 & \cellcolor{shadeblue!28}0.157 & \cellcolor{shadeblue!28}0.206 & \cellcolor{accentred!28}+0.047 & \cellcolor{shadeblue!8}18.4 & \cellcolor{shadeblue!20}4.8 & \cellcolor{shadeblue!28}0.961 \\
\bottomrule
\end{tabular}
\end{center}
\end{table}

\section{Allocation Dynamics and Headroom}
\label{app:dynamics}

Figure~\ref{fig:cumulative} shows mean cumulative RVC against tokens used at budget 1024 on the WVS protected fold and the ESS10 held-out set under this version's full-pool protocol: on WVS, MMR covers faster through the first half of the budget, but LMVU overtakes it at roughly 530 tokens and closes the entire gap to the full-pool oracle ($1.000$ vs.\ $0.917$); on ESS10 the two track each other for the first $\sim$250 tokens, after which LMVU leads throughout ($0.889$ vs.\ $0.843$ at the budget).

\begin{figure}[h]
\begin{center}
\includegraphics[width=\linewidth]{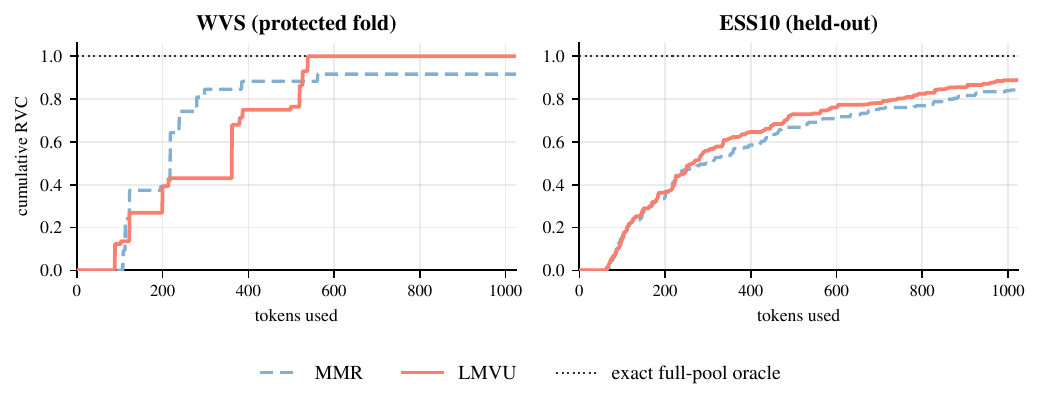}
\end{center}
\caption{Mean cumulative RVC vs.\ tokens used at budget 1024 under the full-pool protocol (WVS protected fold; ESS10 held-out): the LMVU local-training arm (three seeds, seed means within query, macro-averaged) against deterministic full-pool MMR; the dotted line represents the exact full-pool oracle, which has RVC equal to 1.0 on each of the two benchmarks.}
\label{fig:cumulative}
\end{figure}

The sweep of Figure~\ref{fig:budgets} shows that the gain is not tied to one operating point: on ESS10 LMVU improves over MMR by $+0.047$ at 1024 and $+0.040$ at 2048; on Perspectrum LMVU leads MMR at all three budgets while M0 remains strongest at 1024; and on ArgKP the only nonzero gap is $+0.0095$ at 512. The descriptive correlation (Pearson $r{=}0.574$, $n=4$) links realized gains to available headroom rather than to budget alone.

\section{Repeated-Support Coverage}
\label{app:rvc2}

This appendix collects the full repeated-support study behind \S\ref{sec:problem}, \S\ref{sec:experiments}, and Table~\ref{tab:primary}(a,b): the source-unit definitions, the audit protocol, the complete budget grid, and the post-hoc objective-alignment ablation that trains on the repeated-support target directly.

\paragraph{The $\RVCn$ family.}
For the evaluation family in \S\ref{sec:experiments}, $\RVC$ credits a viewpoint once, whereas higher thresholds retain credit for additional distinct source units. Writing the source count explicitly,
\[
\RVCn(S) \;=\; \frac{1}{|V_q|} \sum_{v \in V_q} \min\Bigl\{\frac{n_v(S)}{n},\, 1\Bigr\},
\qquad
n_v(S) \;=\; \bigl|\{\,\mathrm{src}(e) : e \in S,\ v \in V(e)\,\}\bigr|,
\]
where $\mathrm{src}(e)$ maps an evidence item to its source unit. $\RVCn$ is monotone in $S$ and nonincreasing in $n$, and $n{=}1$ recovers $\RVC$ exactly. We report $n{=}1$ (breadth) and $n{=}2$, the minimal extension beyond singleton coverage: $\RVCtwo$ is one precisely when every valid viewpoint is backed by at least two distinct sources. The repeated-support ceiling replaces the oracle objective by $\sum_{v \in V_q} \min\{n_v(S), 2\}$ and is solved exactly by a base-3 (trit-state) dynamic program over per-viewpoint source counts on the full selection pools $E(q)$ (Proposition~\ref{prop:oracle2}), giving $S_q^{\star(2)}$ and $\RVCtwo(S_q^{\star(2)})$, with $\mathrm{headroom}_2 = \RVCtwo(S_q^{\star(2)}) - \RVCtwo(\text{MMR})$ defined analogously to headroom.

\paragraph{Source units.}
Table~\ref{tab:source_units} lists what $n_v(S)$ counts distinct instances of on each benchmark: a survey respondent on WVS and ESS10, one claim-scoped document on Perspectrum, and the individual argument itself on ArgKP. Parsing is total---every item in the full pools maps to exactly one unit ($0$ parse failures, verified per query--budget cell in the oracle audit)---and within every query distinct pool items carry distinct units, the injectivity hypothesis of Proposition~\ref{prop:oracle2}. Two cautions apply. First, repeated support is not independent corroboration: Perspectrum's documents trace to only three source websites, so two documents supporting one viewpoint may share an origin; we report the document-level unit as primary and make no site-independence claim. Second, $\RVCtwo$ is bounded by $\RVC$ tra\-jectory-wise (a second unit implies a first) and held on every scored trajectory.

\begin{table}[h]
\caption{Source units for $\RVCtwo$ per benchmark: what $n_v(S)$ counts distinct instances of. Pool items and distinct units are totals over the full selection pools; within each query the item-to-unit mapping is injective. Shading scales with value within column.}
\label{tab:source_units}
\begin{center}
\small
\begin{tabular}{llcc}
\toprule
Benchmark & Source unit & Pool items & Distinct units \\
\midrule
WVS (protected) & survey respondent & \cellcolor{shadeblue!28}3044 & \cellcolor{shadeblue!16}1522 \\
Perspectrum & claim-scoped document & \cellcolor{shadeblue!8}859 & \cellcolor{shadeblue!8}859 \\
ArgKP & individual argument & \cellcolor{shadeblue!4}454 & \cellcolor{shadeblue!4}454 \\
ESS10 & survey respondent & \cellcolor{shadeblue!28}3003 & \cellcolor{shadeblue!28}2290 \\
\bottomrule
\end{tabular}
\end{center}
\end{table}

\paragraph{Audit protocol of the full-pool extension.}
All trajectories come from two gated sources (Appendix~\ref{app:protocol}): the delivered LF matrix (integrity-verified, bit-reproduced by parameter replay) and the full-pool baseline rerun (gated on per-query reproduction of the delivered MMR and Dense rows at tolerance $10^{-9}$). Every run carries both $\RVC$ and $\RVCtwo$ under the source-unit definitions above. The exact repeated-support oracle was computed once per query--budget cell on the full pools and verified by an in-process identical second run; it dominates every method's $\RVCtwo$ on every query ($0$ violations) and never exceeds the breadth oracle. As a cross-check, the same construction restricted to the original local pools reproduces the locked local-pool repeated-support audit exactly on all twelve dataset--budget cells. For the protected WVS fold this audit reused the same gold annotations read-only; no method, tuning, or training state changed.

\paragraph{Full budget grid.}
Table~\ref{tab:rvc2_full} reports $\RVCtwo$ for all methods and the exact repeated-support oracle under the full-pool protocol. WVS is tabulated at the primary budget only, matching the main-text reporting scope; the other benchmarks span all three budgets.

\begin{table}[!ht]
\caption{Repeated-support coverage $\RVCtwo$ across budgets (macro mean per query, same protocol as Table~\ref{tab:primary}); WVS at the primary budget only. Best non-oracle per column in bold; shading scales with $\RVCtwo$ within column (oracle row excluded). Yellow rows are the training-heavy neural baselines (JPR, RichRAG); they are excluded from the bold comparison and from the per-column shading scales.}
\label{tab:rvc2_full}
\begin{center}
\scriptsize
\renewcommand{\arraystretch}{0.88}
\setlength{\tabcolsep}{2.6pt}
\begin{tabular}{lcccccccccc}
\toprule
& WVS & \multicolumn{3}{c}{Perspectrum} & \multicolumn{3}{c}{ArgKP} & \multicolumn{3}{c}{ESS10} \\
\cmidrule(lr){2-2}\cmidrule(lr){3-5}\cmidrule(lr){6-8}\cmidrule(lr){9-11}
& 1024 & 512 & 1024 & 2048 & 512 & 1024 & 2048 & 512 & 1024 & 2048 \\
\midrule
Random & \cellcolor{shadeblue!28}0.8757 & \cellcolor{shadeblue!24}0.5296 & \cellcolor{shadeblue!24}0.7932 & \cellcolor{shadeblue!24}0.9491 & \cellcolor{shadeblue!20}0.8651 & \cellcolor{shadeblue!20}0.9720 & \cellcolor{shadeblue!28}1.0000 & \cellcolor{shadeblue!24}0.4935 & \cellcolor{shadeblue!24}0.7223 & \cellcolor{shadeblue!28}0.8975 \\
BM25 & \cellcolor{shadeblue!16}0.6583 & \cellcolor{shadeblue!20}0.5083 & \cellcolor{shadeblue!24}0.7720 & \cellcolor{shadeblue!24}0.9276 & \cellcolor{shadeblue!8}0.7952 & \cellcolor{shadeblue!24}0.9881 & \cellcolor{shadeblue!28}1.0000 & \cellcolor{shadeblue!24}0.4827 & \cellcolor{shadeblue!24}0.6769 & \cellcolor{shadeblue!24}0.8657 \\
FacLoc & \cellcolor{shadeblue!28}\textbf{0.9167} & \cellcolor{shadeblue!24}0.5370 & \cellcolor{shadeblue!24}0.8045 & \cellcolor{shadeblue!24}0.9419 & \cellcolor{shadeblue!8}0.7863 & \cellcolor{shadeblue!24}0.9881 & \cellcolor{shadeblue!28}1.0000 & \cellcolor{shadeblue!24}0.5036 & \cellcolor{shadeblue!24}0.6781 & \cellcolor{shadeblue!24}0.8583 \\
xQuAD & \cellcolor{shadeblue!20}0.7114 & \cellcolor{shadeblue!24}0.5477 & \cellcolor{shadeblue!24}0.8083 & \cellcolor{shadeblue!24}0.9426 & \cellcolor{shadeblue!16}0.8450 & \cellcolor{shadeblue!20}0.9733 & \cellcolor{shadeblue!28}1.0000 & \cellcolor{shadeblue!24}0.4474 & \cellcolor{shadeblue!24}0.6422 & \cellcolor{shadeblue!24}0.8348 \\
kDPP & \cellcolor{shadeblue!28}0.9042 & \cellcolor{shadeblue!24}0.5518 & \cellcolor{shadeblue!24}0.7995 & \cellcolor{shadeblue!24}0.9503 & \cellcolor{shadeblue!24}0.9143 & \cellcolor{shadeblue!28}1.0000 & \cellcolor{shadeblue!28}1.0000 & \cellcolor{shadeblue!24}0.5043 & \cellcolor{shadeblue!24}0.7047 & \cellcolor{shadeblue!28}0.9032 \\
SubMix & \cellcolor{shadeblue!20}0.6898 & \cellcolor{shadeblue!24}0.5404 & \cellcolor{shadeblue!24}0.8002 & \cellcolor{shadeblue!24}0.9397 & \cellcolor{shadeblue!16}0.8310 & \cellcolor{shadeblue!20}0.9714 & \cellcolor{shadeblue!28}1.0000 & \cellcolor{shadeblue!24}0.4515 & \cellcolor{shadeblue!24}0.6173 & \cellcolor{shadeblue!24}0.8283 \\
LexRank & \cellcolor{shadeblue!20}0.7315 & \cellcolor{shadeblue!24}0.5411 & \cellcolor{shadeblue!24}0.7986 & \cellcolor{shadeblue!24}0.9389 & \cellcolor{shadeblue!16}0.8310 & \cellcolor{shadeblue!24}0.9881 & \cellcolor{shadeblue!28}1.0000 & \cellcolor{shadeblue!24}0.4566 & \cellcolor{shadeblue!24}0.6220 & \cellcolor{shadeblue!24}0.8344 \\
AdaGReS & \cellcolor{shadeblue!16}0.6148 & \cellcolor{shadeblue!4}0.2904 & \cellcolor{shadeblue!4}0.5018 & \cellcolor{shadeblue!4}0.7692 & \cellcolor{shadeblue!4}0.7548 & \cellcolor{shadeblue!4}0.9143 & \cellcolor{shadeblue!28}1.0000 & \cellcolor{shadeblue!16}0.3058 & \cellcolor{shadeblue!20}0.5138 & \cellcolor{shadeblue!20}0.7026 \\
MMKP & \cellcolor{shadeblue!4}0.3120 & \cellcolor{shadeblue!28}0.6078 & \cellcolor{shadeblue!28}0.8329 & \cellcolor{shadeblue!24}0.9465 & \cellcolor{shadeblue!20}0.8649 & \cellcolor{shadeblue!4}0.9173 & \cellcolor{shadeblue!4}0.9173 & \cellcolor{shadeblue!4}0.1251 & \cellcolor{shadeblue!4}0.1251 & \cellcolor{shadeblue!4}0.1251 \\
\midrule
Dense & \cellcolor{shadeblue!20}0.6898 & \cellcolor{shadeblue!24}0.5432 & \cellcolor{shadeblue!24}0.8014 & \cellcolor{shadeblue!24}0.9397 & \cellcolor{shadeblue!16}0.8310 & \cellcolor{shadeblue!20}0.9714 & \cellcolor{shadeblue!28}1.0000 & \cellcolor{shadeblue!24}0.4515 & \cellcolor{shadeblue!20}0.6154 & \cellcolor{shadeblue!24}0.8283 \\
CeRerank & \cellcolor{shadeblue!20}0.6704 & \cellcolor{shadeblue!20}0.5230 & \cellcolor{shadeblue!24}0.7923 & \cellcolor{shadeblue!24}0.9353 & \cellcolor{shadeblue!8}0.7774 & \cellcolor{shadeblue!12}0.9458 & \cellcolor{shadeblue!28}1.0000 & \cellcolor{shadeblue!20}0.3933 & \cellcolor{shadeblue!20}0.5285 & \cellcolor{shadeblue!20}0.7149 \\
Cluster & \cellcolor{shadeblue!24}0.7753 & \cellcolor{shadeblue!24}0.5488 & \cellcolor{shadeblue!24}0.7999 & \cellcolor{shadeblue!28}0.9508 & \cellcolor{shadeblue!16}0.8512 & \cellcolor{shadeblue!16}0.9631 & \cellcolor{shadeblue!28}1.0000 & \cellcolor{shadeblue!24}0.4710 & \cellcolor{shadeblue!24}0.6960 & \cellcolor{shadeblue!28}0.8781 \\
MMR & \cellcolor{shadeblue!24}0.7824 & \cellcolor{shadeblue!24}0.5471 & \cellcolor{shadeblue!24}0.8160 & \cellcolor{shadeblue!24}0.9500 & \cellcolor{shadeblue!24}0.8929 & \cellcolor{shadeblue!28}1.0000 & \cellcolor{shadeblue!28}1.0000 & \cellcolor{shadeblue!28}0.5312 & \cellcolor{shadeblue!28}0.7273 & \cellcolor{shadeblue!28}0.8778 \\
M0 & \cellcolor{shadeblue!20}0.7444 & \cellcolor{shadeblue!28}\textbf{0.6080} & \cellcolor{shadeblue!28}\textbf{0.8480} & \cellcolor{shadeblue!28}\textbf{0.9668} & \cellcolor{shadeblue!12}0.7988 & \cellcolor{shadeblue!4}0.9131 & \cellcolor{shadeblue!28}1.0000 & \cellcolor{shadeblue!24}0.4980 & \cellcolor{shadeblue!24}0.7206 & \cellcolor{shadeblue!24}0.8559 \\
\cellcolor{accentred!18}LMVU & \cellcolor{shadeblue!28}0.9015 & \cellcolor{shadeblue!24}0.5753 & \cellcolor{shadeblue!28}0.8373 & \cellcolor{shadeblue!28}0.9640 & \cellcolor{shadeblue!28}\textbf{0.9310} & \cellcolor{shadeblue!28}\textbf{1.0000} & \cellcolor{shadeblue!28}\textbf{1.0000} & \cellcolor{shadeblue!28}\textbf{0.5532} & \cellcolor{shadeblue!28}\textbf{0.7790} & \cellcolor{shadeblue!28}\textbf{0.9343} \\
\midrule
JPR & \cellcolor{shadeyellow!24}0.8097 & \cellcolor{shadeyellow!24}0.5402 & \cellcolor{shadeyellow!24}0.7996 & \cellcolor{shadeyellow!28}0.9548 & \cellcolor{shadeyellow!28}0.9921 & \cellcolor{shadeyellow!28}1.0000 & \cellcolor{shadeyellow!28}1.0000 & \cellcolor{shadeyellow!28}0.6075 & \cellcolor{shadeyellow!24}0.6839 & \cellcolor{shadeyellow!24}0.7511 \\
RichRAG & \cellcolor{shadeyellow!24}0.7994 & \cellcolor{shadeyellow!20}0.5151 & \cellcolor{shadeyellow!24}0.7906 & \cellcolor{shadeyellow!24}0.9302 & \cellcolor{shadeyellow!8}0.7972 & \cellcolor{shadeyellow!8}0.9234 & \cellcolor{shadeyellow!28}1.0000 & \cellcolor{shadeyellow!28}0.5442 & \cellcolor{shadeyellow!28}0.8315 & \cellcolor{shadeyellow!28}0.9674 \\
\midrule
Oracle$_2$ & \cellcolor{oraclegray}1.0000 & \cellcolor{oraclegray}0.6320 & \cellcolor{oraclegray}0.8898 & \cellcolor{oraclegray}0.9907 & \cellcolor{oraclegray}1.0000 & \cellcolor{oraclegray}1.0000 & \cellcolor{oraclegray}1.0000 & \cellcolor{oraclegray}0.6612 & \cellcolor{oraclegray}0.9334 & \cellcolor{oraclegray}1.0000 \\
\bottomrule
\end{tabular}
\end{center}
\end{table}

The strongest depth allocator depends on the benchmark: M0 leads Perspectrum at all three budgets, FacLoc edges LMVU on WVS at 1024 ($0.9167$ vs.\ $0.9015$), and on ESS10 LMVU leads the non-neural field throughout the grid, while the training-heavy \textsc{RichRAG} (yellow row, outside the bold comparison) surpasses it at 1024 and 2048 ($0.8315$ vs.\ $0.7790$; $0.9674$ vs.\ $0.9343$) but not at 512 ($0.5442$ vs.\ $0.5532$). Budget changes also have different effects. LMVU's $\RVCtwo$ advantage over MMR narrows from $0.0282$ to $0.0140$ on Perspectrum but grows from $0.0220$ to $0.0565$ on ESS10. Higher absolute coverage need not mean a smaller gap to the achievable ceiling: on ESS10, the remaining Oracle$_2$--LMVU gap is $0.108$, $0.154$, and $0.066$ at 512, 1024, and 2048 tokens. Unlike the MMR-based $\mathrm{headroom}_2$, this gap measures LMVU's remaining shortfall against a budget-dependent ceiling.

Breadth saturation can conceal incomplete repeated support. On WVS at 1024, LMVU reaches RVC $1.000$ but $\RVCtwo$ only $0.9015$, below Oracle$_2$ $1.000$. ArgKP is not fully breadth-saturated at 512: LMVU and MMR attain RVC $0.9762$ and $0.9667$, respectively. Its depth gap is larger (Oracle$_2$ $1.000$ vs.\ MMR $0.893$), and LMVU captures part ($0.931$, $+0.038$ over MMR). Both LMVU and MMR reach the breadth and depth ceilings at 1024 and 2048, explaining their zero differences at those budgets rather than establishing equivalence away from saturation.

The mean depth deficits at 512 tokens on WVS and ESS10 under the original local-selection protocol are absent from these full-pool comparisons. LMVU improves both metrics over MMR on ESS10 ($0.553$ vs.\ $0.531$ for $\RVCtwo$; Table~\ref{tab:rvc2_full}). The same depth ordering holds on WVS ($0.735$ vs.\ $0.698$), an existing 512-token result not tabulated here; its breadth gain is also positive. These observations describe the full-pool runs, not a new protected-fold evaluation or a general guarantee that breadth and depth objectives agree across datasets and budgets.

\paragraph{Paired uncertainty.}
Table~\ref{tab:rvc2_budgets} reports paired LMVU--MMR differences for the same tabulated budgets. These exploratory statistics distinguish positive macro means from consistent query-level gains, retaining ties and losses rather than treating every query as improved.

\begin{table}[H]
\caption{LMVU $-$ MMR paired statistics: $\Delta\RVC$ at the primary budget 1024 (top block) and $\Delta\RVCtwo$ across budgets (bottom block); WVS at the primary budget only. Intervals are bootstrap 95\% CIs, $p$ the two-sided sign test; WVS/ArgKP statistics are descriptive (three groups). Positive $\Delta$ cells shaded by magnitude.}
\label{tab:rvc2_budgets}
\begin{center}
\small
\renewcommand{\arraystretch}{0.88}
\begin{tabular}{lllcccc}
\toprule
Metric & Benchmark & Budget & $\Delta$ & 95\% CI & w/t/l & $p$ \\
\midrule

\multirow{4}{*}{$\RVC$}
& WVS         & 1024 & \cellcolor{accentred!28}+0.0833 & ---                     & 2/4/0   & 1.00 \\
& Perspectrum & 1024 & \cellcolor{accentred!8}+0.0161  & [$-0.001$, $+0.035$]    & 11/82/4 & 0.12 \\
& ArgKP       & 1024 & 0.0000                            & ---                     & 0/6/0   & 1.00 \\
& ESS10       & 1024 & \cellcolor{accentred!16}+0.0466 & [$-0.022$, $+0.118$]    & 10/18/6 & 0.45 \\

\midrule

\multirow{10}{*}{$\RVCtwo$}
& WVS         & 1024 & \cellcolor{accentred!28}+0.1191 & ---                     & 4/1/1   & 1.00 \\

\cmidrule(lr){2-7}
& \multirow{3}{*}{Perspectrum}
              & 512  & \cellcolor{accentred!8}+0.0282  & [$+0.014$, $+0.044$]    & 19/76/2 & 0.0002 \\
&             & 1024 & \cellcolor{accentred!8}+0.0213  & [$+0.007$, $+0.037$]    & 16/75/6 & 0.052 \\
&             & 2048 & \cellcolor{accentred!8}+0.0140  & [$+0.004$, $+0.025$]    & 11/83/3 & 0.057 \\

\cmidrule(lr){2-7}
& \multirow{3}{*}{ArgKP}
              & 512  & \cellcolor{accentred!16}+0.0381 & ---                     & 3/2/1   & 1.00 \\
&             & 1024 & 0.0000                           & ---                     & 0/6/0   & 1.00 \\
&             & 2048 & 0.0000                           & ---                     & 0/6/0   & 1.00 \\

\cmidrule(lr){2-7}
& \multirow{3}{*}{ESS10}
              & 512  & \cellcolor{accentred!8}+0.0220  & [$-0.016$, $+0.065$]    & 7/22/5  & 0.77 \\
&             & 1024 & \cellcolor{accentred!28}+0.0518 & [$-0.000$, $+0.109$]    & 15/11/8 & 0.21 \\
&             & 2048 & \cellcolor{accentred!28}+0.0565 & [$+0.009$, $+0.107$]    & 12/17/5 & 0.14 \\

\bottomrule
\end{tabular}
\end{center}
\end{table}

The breadth block at the primary budget is nonnegative everywhere (ArgKP ties at the ceiling) and no interval excludes zero. On Perspectrum, all three bootstrap intervals exclude zero, but only the 512-token sign test has $p<0.05$; at 2048, 83 of 97 queries tie. On ESS10, all mean differences are positive, yet only the 2048-token bootstrap interval excludes zero and none of the sign tests has $p<0.05$. The 1024-token lower bound is slightly negative, displayed as $-0.000$ after rounding. These summaries answer different questions: the bootstrap interval retains the magnitudes of query-level differences, whereas the sign test uses their directions after excluding ties. Thus an interval above zero need not imply that wins substantially outnumber losses among the queries that change.

For WVS and ArgKP, the reported win/tie/loss counts describe six queries, but the descriptive sign tests operate on three issue or topic groups. The group count is too small for the CI reporting used on the other benchmarks. In particular, WVS's positive mean gain is not evidence that every protected query improves on depth. Across benchmarks, positive means therefore coexist with unchanged queries and occasional regressions under the same budget.

\paragraph{Objective-alignment ablation (post hoc).}
We next isolate the training target under the original local-selection protocol (Appendix~\ref{app:protocol}), separately from the full-pool results above. LMVU's training target is the breadth gain $\Delta V$; a natural question is whether pointing the same construction at depth helps $\RVCtwo$ directly. The repeated-support realized gain of a candidate in state $(q, S)$ is
\[
\Delta_2(e \mid q, S) \;=\; \frac{1}{2\,|V_q|}\; \#\bigl\{\, v \in V(e) \,:\; n_v(S) < 2 \ \text{and}\ \mathrm{src}(e) \notin \mathrm{sources}_v(S) \,\bigr\},
\]
exactly the realized increment of $\RVCtwo$ (a viewpoint gains $1/(2|V_q|)$ for its first and again for its second distinct supporting unit). We retrained the unchanged ranker class (linear, the same 18 features, $C{=}1.0$, 20-pair cap, seed 41) on the canonical development states with $\Delta_2$ targets on ESS10 and Perspectrum, and evaluated the frozen model on the same held-out folds. The protected WVS fold was untouched; ArgKP was excluded as saturated at the primary budget. This post-hoc arm is an objective-alignment probe, not a candidate method; no claim rests on it.

{\setlength{\intextsep}{6pt}
\begin{table}[H]
\caption{Objective-alignment ablation: the canonical M1 ranker vs.\ the same construction trained on $\Delta_2$ (M1$_{\Delta_2}$), with paired $\Delta\RVCtwo$ statistics (bootstrap 95\% CI; sign-test $p$). Post-hoc arm on held-out folds under the original protocol (Appendix~\ref{app:protocol}); positive $\Delta$ cells shaded by magnitude.}
\label{tab:delta2}
\begin{center}
\scriptsize
\renewcommand{\arraystretch}{0.88}
\setlength{\tabcolsep}{2pt}
\begin{tabular}{llccccccc}
\toprule
& & \multicolumn{3}{c}{RVC} & \multicolumn{4}{c}{$\RVCtwo$} \\
\cmidrule(lr){3-5}\cmidrule(lr){6-9}
Benchmark & Budget & M1 & M1$_{\Delta_2}$ & $\Delta$ [95\% CI]
& M1 & M1$_{\Delta_2}$ & $\Delta$ [95\% CI] & $p$ \\
\midrule

\multirow{3}{*}{Perspectrum}
& 512
& 0.8134 & 0.8292
& \cellcolor{accentred!22}$+0.0157$ [$-0.007$, $+0.041$]
& 0.5751 & 0.6043
& \cellcolor{accentred!28}$+0.0292$ [$+0.013$, $+0.047$]
& 0.001 \\

& 1024
& 0.9392 & 0.9563
& \cellcolor{accentred!22}$+0.0172$ [$+0.002$, $+0.034$]
& 0.8286 & 0.8509
& \cellcolor{accentred!28}$+0.0224$ [$+0.004$, $+0.041$]
& 0.004 \\

& 2048
& 0.9901 & 0.9933
& \cellcolor{accentred!16}$+0.0032$ [$-0.006$, $+0.013$]
& 0.9628 & 0.9678
& \cellcolor{accentred!16}$+0.0051$ [$-0.006$, $+0.017$]
& 0.774 \\

\cmidrule(lr){2-9}

\multirow{3}{*}{ESS10}
& 512
& 0.6707 & 0.6962
& \cellcolor{accentred!28}$+0.0255$ [$-0.029$, $+0.083$]
& 0.5170 & 0.5272
& \cellcolor{accentred!22}$+0.0102$ [$-0.019$, $+0.039$]
& 0.388 \\

& 1024
& 0.8866 & 0.8678
& $-0.0187$ [$-0.051$, $+0.012$]
& 0.7742 & 0.7610
& $-0.0133$ [$-0.038$, $+0.009$]
& 0.549 \\

& 2048
& 0.9739 & 0.9815
& \cellcolor{accentred!16}$+0.0076$ [$-0.010$, $+0.027$]
& 0.9284 & 0.9348
& \cellcolor{accentred!16}$+0.0064$ [$-0.014$, $+0.026$]
& 0.549 \\

\bottomrule
\end{tabular}
\end{center}
\end{table}

Table~\ref{tab:delta2} shows that target alignment helps under some, but not all, conditions. On Perspectrum, $\Delta_2$ training improves $\RVCtwo$ by $+0.029$ at 512 and $+0.022$ at 1024, with CIs excluding zero and sign-test $p = 0.001$, $0.004$. At 1024, RVC also increases by $+0.017$, with a positive bootstrap interval: improved depth need not come at the expense of breadth in this setting. At 2048, the smaller depth gain of $+0.005$ has an interval crossing zero, so the lower-budget evidence does not establish a reliable benefit throughout the budget range.

On ESS10, all three depth intervals cross zero. At 1024, both RVC and $\RVCtwo$ decline in mean ($-0.019$ and $-0.013$), with neither interval excluding zero. This is not an observed tradeoff in which depth improves while breadth deteriorates; nor does it identify redundancy as a causal explanation. The experiment changes the target, but does not guarantee that the same ranker will learn a better selection policy for that target. These comparisons concern canonical M1 and M1$_{\Delta_2}$ within the original protocol, not the full-pool LMVU scores in Table~\ref{tab:rvc2_full}.

Together, the full-pool evaluation and the separate target probe show why breadth and repeated support should be reported together. A breadth ceiling can leave depth unresolved, while directly training for depth does not yield a uniform improvement across datasets and budgets. These are selection-level findings: distinct source units need not provide independent corroboration, and neither study establishes an improvement in end-to-end generation quality.

%% file: sections/analysis.tex
\section{Formal Analysis}
\label{app:analysis}

This section collects the formal statements behind the assertions of \S\ref{sec:problem} and \S\ref{sec:method}. We reuse the notation established there: a query $q$ carries a viewpoint inventory $V_q$; each evidence item $e$ covers $V(e) \subseteq V_q$ at token cost $c(e)$; the coverage value of a selection is $F_q(S) = |V(S)|$ with $V(S) = \bigcup_{e \in S} V(e)$, so that $\RVC(S) = F_q(S)/|V_q|$; the budget is $B$; and the realized marginal gain of a candidate in state $(q, S)$ is $\Delta V(e \mid q, S) = F_q(S \cup \{e\}) - F_q(S)$.

\paragraph{The coverage objective.}

\begin{proposition}[Monotonicity and submodularity]
\label{prop:submod}
For every query $q$, the coverage function $F_q$ is monotone and submodular over all subsets of the full evidence pool $E(q)$.
\end{proposition}
\begin{proof}
Monotonicity: $A \subseteq B$ implies $V(A) \subseteq V(B)$, hence $F_q(A) \le F_q(B)$. Diminishing returns: for $A \subseteq B \subseteq E(q)$ and $e \notin B$,
\[
\Delta V(e \mid q, A) \;=\; |V(e) \setminus V(A)| \;\ge\; |V(e) \setminus V(B)| \;=\; \Delta V(e \mid q, B),
\]
where $V(A) \subseteq V(B)$ ensures that enlarging the selection cannot increase the marginal gain.
\end{proof}

\paragraph{The full-pool oracle is exact.}

\begin{proposition}[Exactness and cost of the bitmask oracle]
\label{prop:oracle}
Let $m_q = |V_q|$ and fix an enumeration of $V_q$, so each $V(e)$ is an $m_q$-bit mask. For every mask $\mu$, the subset dynamic program stores the minimum total token cost of a sub-selection whose covered mask equals $\mu$. After all pool items have been processed, the best budget-feasible entry in this table gives the exact value $\max_{S \subseteq E(q),\, c(S) \le B} F_q(S)$. The computation uses $O\bigl(|E(q)| \cdot 2^{m_q}\bigr)$ transitions to process the items and $O(2^{m_q})$ memory for the table of minimum costs.
\end{proposition}
\begin{proof}
Process the items $e_1, \dots, e_n$ of $E(q)$ in any order and maintain $\mathrm{dp}_i[\mu] =$ the minimum cost of a subset of $\{e_1, \dots, e_i\}$ covering exactly mask $\mu$ ($+\infty$ if unattainable). The transitions $\mathrm{dp}_{i+1}[\mu] \leftarrow \min\bigl(\mathrm{dp}_{i+1}[\mu],\, \mathrm{dp}_i[\mu]\bigr)$ (skip) and $\mathrm{dp}_{i+1}[\mu \mathbin{\mathrm{OR}} V(e_{i+1})] \leftarrow \min\bigl(\mathrm{dp}_{i+1}[\mu \mathbin{\mathrm{OR}} V(e_{i+1})],\, \mathrm{dp}_i[\mu] + c(e_{i+1})\bigr)$ (take) preserve the invariant: by induction, every subset of the first $i{+}1$ items either excludes $e_{i+1}$ or includes it, and each case is costed optimally by the corresponding transition. Since $F_q$ depends on $S$ only through its covered mask, the optimum is $\max\{\, |\mu| : \mathrm{dp}_n[\mu] \le B \,\}$. There are $2^{m_q}$ masks and each item triggers one transition per mask.
\end{proof}

\noindent On all four benchmarks $m_q \le 9$ (viewpoint inventories, unchanged by the pool), so a query's table has at most $2^9 = 512$ mask states: the exact per-query solve on the full pool is negligible, which is why the paper reports exact rather than approximate oracle values.

\paragraph{The repeated-support oracle is exact.}

\begin{proposition}[Exactness and cost of the trit-state oracle]
\label{prop:oracle2}
Let $m_q = |V_q|$ and fix an enumeration of $V_q$ and a source map $\mathrm{src}$. Assume that within the query, distinct pool items carry distinct source units---verified on all four benchmarks, where the source unit is the respondent (WVS, ESS10), the document itself (Perspectrum), or the item itself (ArgKP). The dynamic program that stores, for every vector $\tau \in \{0,1,2\}^{m_q}$ of capped per-viewpoint source counts, the minimum total token cost of a sub-selection realizing exactly $\tau$ returns $\max_{S \subseteq E(q),\, c(S) \le B} \sum_{v} \min\{n_v(S), 2\}$ exactly, using $O\bigl(|E(q)| \cdot 3^{m_q}\bigr)$ transitions and $O(3^{m_q})$ memory. The pool may first be preprocessed by keeping only a cheapest item for each pair $(\pi, s)$ of support pattern $\pi \subseteq V_q$ and source $s$. This reduction preserves the optimal value of the repeated-support objective.
\end{proposition}
\begin{proof}
Encode $\tau$ as a base-3 integer and maintain $\mathrm{dp}_i[\tau] =$ the minimum cost of a subset of the first $i$ items whose capped source-count vector is exactly $\tau$ ($+\infty$ if unattainable). Taking item $e_{i+1}$ with pattern $\pi$ and source $s$ maps $\tau$ to $\tau'$ with $\tau'_v = \min\{\tau_v + 1, 2\}$ for $v \in \pi$ and $\tau'_v = \tau_v$ otherwise. Under the injectivity hypothesis, each taken item introduces a source unseen by any other taken item, so this per-item increment is exactly the increment of capped \emph{distinct-source} counts. The skip/take transitions preserve the min-cost invariant by the same induction as Proposition~\ref{prop:oracle}, and the objective depends on $S$ only through $\tau$, so the optimum is $\max\{\,\textstyle\sum_v \tau_v : \mathrm{dp}_n[\tau] \le B\,\}$. For the preprocessing: two items sharing both pattern $\pi$ and source $s$ contribute the identical increment, so any selection containing both can replace the costlier one by the cheaper one (or drop it) without changing the realized $\tau$; an optimum therefore survives among the retained items.
.\end{proof}

\noindent With $m_q \le 9$ the table holds at most $3^9 = 19{,}683$ states. Exact repeated-support optimization therefore remains inexpensive for the viewpoint inventories evaluated here.

\paragraph{Why raw-gain greedy has no knapsack guarantee.}

\begin{proposition}[Unbounded approximation ratio under a knapsack budget]
\label{prop:knapsack}
For every $\varepsilon > 0$, there is an instance of the budgeted viewpoint-coverage problem on which budget-feasible raw-gain greedy returns a selection $S^{\mathrm{gr}}$ with $F_q(S^{\mathrm{gr}}) / F_q(S_q^\star) < \varepsilon$. Here raw-gain greedy denotes the loop of Algorithm~\ref{alg:lmvu} scored by the true gains $\Delta V$.
\end{proposition}
\begin{proof}
Fix an integer $m \ge 3$ and let $V_q = \{v_1, \dots, v_{m+1}\}$ and $B = m$. The pool contains an item $a$ with $c(a) = m$ and $V(a) = \{v_1, v_2\}$, and items $b_1, \dots, b_{m+1}$ with $c(b_i) = 1$ and $V(b_i) = \{v_i\}$. Raw-gain greedy first compares $\Delta V(a \mid q, \emptyset) = 2$ against $\Delta V(b_i \mid q, \emptyset) = 1$ and selects $a$, which exhausts the budget, so $F_q(S^{\mathrm{gr}}) = 2$. The optimum selects any $m$ of the $b_i$, giving $F_q(S_q^\star) = m$. The ratio $2/m$ falls below $\varepsilon$ for all $m \ge \lceil 2/\varepsilon \rceil$.
\end{proof}

\begin{remark}
Under equal unit costs (a cardinality constraint) the same raw-gain greedy rule recovers the classical $(1 - 1/e)$ guarantee \citep{nemhauser1978analysis}, and under a knapsack budget $(1 - 1/e)$ is restored by density-aware partial-enumeration variants \citep{khuller1999budgeted,sviridenko2004note}. Every such rule conditions each step on the true marginal gains $\Delta V$---exactly what a label-blind runtime allocator cannot observe (\S\ref{sec:problem}). This is the formal motivation for learning a surrogate gain ordering, and the reason the learned plug-in claims no constant-factor guarantee (\S\ref{sec:method}).
\end{remark}

\paragraph{Surrogate fidelity implies coverage.}

To isolate the role of the learned ordering we analyze an idealized unit-cost case in which every item costs one token, so the budget is the cardinality constraint $k = B$ and the allocator takes exactly $k$ steps.

\begin{proposition}[Per-step shortfall bound]
\label{prop:shortfall}
Suppose unit costs, and let the surrogate-greedy trajectory $S_0 \subsetneq S_1 \subsetneq \dots \subsetneq S_k$ satisfy, at every step $t$, the per-step shortfall condition
\[
\Delta V(e_t \mid q, S_{t-1}) \;\ge\; \max_{e \in E(q) \setminus S_{t-1}} \Delta V(e \mid q, S_{t-1}) \;-\; \eta_t .
\]
Then
\[
F_q(S_k) \;\ge\; \Bigl(1 - \tfrac{1}{e}\Bigr)\, \mathrm{OPT}_k \;-\; \sum_{t=1}^{k} \eta_t,
\qquad \mathrm{OPT}_k \;=\; \max_{|S| \le k} F_q(S).
\]
\end{proposition}
\begin{proof}
Write $F_t = F_q(S_t)$ and let $O$ be an optimal size-$k$ selection. At step $t$,
\[
\begin{aligned}
F_t - F_{t-1}
\;&\ge\; \max_{e} \Delta V(e \mid q, S_{t-1}) - \eta_t
\;\ge\; \tfrac{1}{k} \textstyle\sum_{o \in O} \Delta V(o \mid q, S_{t-1}) - \eta_t \\
\;&\ge\; \tfrac{1}{k}\bigl(F_q(S_{t-1} \cup O) - F_{t-1}\bigr) - \eta_t
\;\ge\; \tfrac{1}{k}\bigl(\mathrm{OPT}_k - F_{t-1}\bigr) - \eta_t,
\end{aligned}
\]
where the second inequality averages the gains over $O$, the third telescopes them by submodularity, and the fourth uses monotonicity. With $\delta_t = \mathrm{OPT}_k - F_t$ this reads $\delta_t \le \bigl(1 - \tfrac{1}{k}\bigr)\, \delta_{t-1} + \eta_t$; unrolling from $\delta_0 \le \mathrm{OPT}_k$ and using $\bigl(1 - \tfrac{1}{k}\bigr)^k \le 1/e$ (and dropping the geometric weights on the $\eta_t$, each at most one) gives $\delta_k \le \mathrm{OPT}_k / e + \sum_{t \le k} \eta_t$.
\end{proof}

\begin{remark}
Three connections back to the paper. (i) The shortfall $\eta_t$ is the quantity that the within-state ranking objective of \S\ref{sec:method:ranking} optimizes a convex surrogate of: a state contributes zero shortfall whenever the model's top-ranked candidate attains the state's maximum realized gain, so the development-fold ranking and learnability metrics (Appendices~\ref{app:dev} and~\ref{app:learn}) are empirical proxies for $\sum_t \eta_t$. (ii) When the within-state orderings are perfect, the surrogate trajectory coincides step-for-step with the trajectory obtained by greedily selecting on true gains. (iii) The cardinality induction requires equal costs and does not apply to the unequal token costs in our experiments. This limitation is consistent with Proposition~\ref{prop:knapsack} and the absence of a constant-factor claim in \S\ref{sec:method}.
\end{remark}

\paragraph{MMR is a fixed-coefficient special case.}

\begin{remark}[MMR as a static restriction of the hypothesis class]
\label{rem:mmr}
The MMR per-step score $\lambda\,\mathrm{rel}(q, e) - (1 - \lambda) \max_{e' \in S} \cos(e, e')$ is a two-term linear form in LMVU's feature vector $\phi(q, S, e)$ (feature groups one and two of Table~\ref{tab:features_full}) with coefficients frozen at $(\lambda, -(1 - \lambda))$---up to per-dataset standardization and positive rescaling, neither of which changes the within-state argmax. The learned hypothesis class of \S\ref{sec:method:model} therefore contains the MMR family; learning $\theta$ replaces the single fixed relevance--novelty trade-off with a state-dependent one, formalizing the remark in \S\ref{sec:method:ranking}.
\end{remark}

%% file: sections/cases.tex
\section{Case Studies}
\label{app:cases}

This section displays trajectories from the full-pool evaluation runs of this version (Appendix~\ref{app:protocol}); gold viewpoint annotations are used to label picks only. Figure~\ref{fig:case_traces} gives the cumulative-coverage view; the tables below list the picks themselves.

\begin{figure}[t]
\begin{center}
\includegraphics[width=0.95\linewidth]{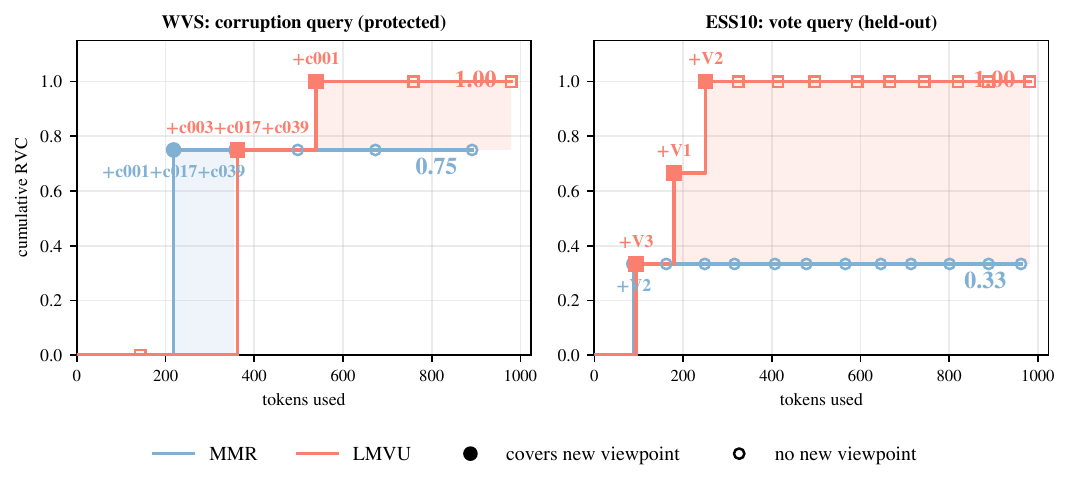}
\end{center}
\caption{Cumulative RVC vs.\ tokens used for the two case-study queries at budget 1024. Filled markers are picks that newly cover at least one viewpoint; open markers are redundant picks; each curve's final RVC is shown in bold.}
\label{fig:case_traces}
\end{figure}

\subsection{WVS protected fold: corruption query}

The corruption query asks how respondents differ in their assessments of corruption across institutions; the annotated inventory has four viewpoints (pervasive institutional corruption, institutions perceived as corruption-free, and two low-perceived-corruption facets). Table~\ref{tab:case_wvs} lists both allocations. MMR opens with a long, high-relevance item (218 tokens) that covers three viewpoints at once, then spends its remaining four picks on items whose viewpoints are already covered; the corruption-free viewpoint (c003) is never selected, leaving RVC at $0.750$ with 891 tokens used. LMVU opens with a short item supporting no annotated viewpoint, then at step~1 selects a 219-token response covering the missing c003 together with c017/c039, and at step~2 adds c001, reaching RVC $1.000$ within 539 tokens; its two remaining picks are redundant, for 979 tokens in total. The divergence is not a relevance decision---MMR's first pick is a good one---it is MMR's inability to discount redundancy against the as-yet-uncovered viewpoint.

\begin{table}[t]
\caption{Pick sequences for the WVS corruption query at budget 1024 (seed 17; the three training seeds coincide on this query): (a)~MMR, (b)~LMVU. Pick: anonymized respondent number (1--10 corruption-perception response). Tokens: per-pick cost. New: newly covered viewpoints (gold, display only). RVC: cumulative. Consecutive picks that newly cover nothing are grouped into one row (step range, pick count, tokens summed). Tags: c001 = pervasive institutional corruption, c003 = institutions perceived as corruption-free, c017/c039 = low perceived institutional corruption. RVC column shaded by magnitude (darker = higher coverage).}
\label{tab:case_wvs}
\begin{center}
\small
\setlength{\tabcolsep}{1.5pt}
\begin{minipage}[t]{0.495\linewidth}
\centering
\textbf{(a) MMR}\par\vspace{2pt}
\begin{tabular}{rlrcc}
\toprule
Step & Pick & Tokens & New & RVC \\
\midrule
0 & 0260 (6/10) & 218 & c001+c017+c039 & \cellcolor{shadeblue!16}0.750 \\
1--4 & 4 picks & 673 & --- & \cellcolor{shadeblue!16}0.750 \\
\bottomrule
\end{tabular}
\end{minipage}\hfill
\begin{minipage}[t]{0.495\linewidth}
\centering
\textbf{(b) LMVU}\par\vspace{2pt}
\begin{tabular}{rlrcc}
\toprule
Step & Pick & Tokens & New & RVC \\
\midrule
0 & 0008 (5/10) & 143 & --- & \cellcolor{shadeblue!4}0.000 \\
1 & 0238 (10/10) & 219 & c003+c017+c039 & \cellcolor{shadeblue!16}0.750 \\
2 & 0280 (10/10) & 177 & c001 & \cellcolor{shadeblue!28}1.000 \\
3--4 & 2 picks & 440 & --- & \cellcolor{shadeblue!28}1.000 \\
\bottomrule
\end{tabular}
\end{minipage}
\end{center}
\end{table}

\begin{figure}[t]
\begin{center}
\includegraphics[width=0.95\linewidth]{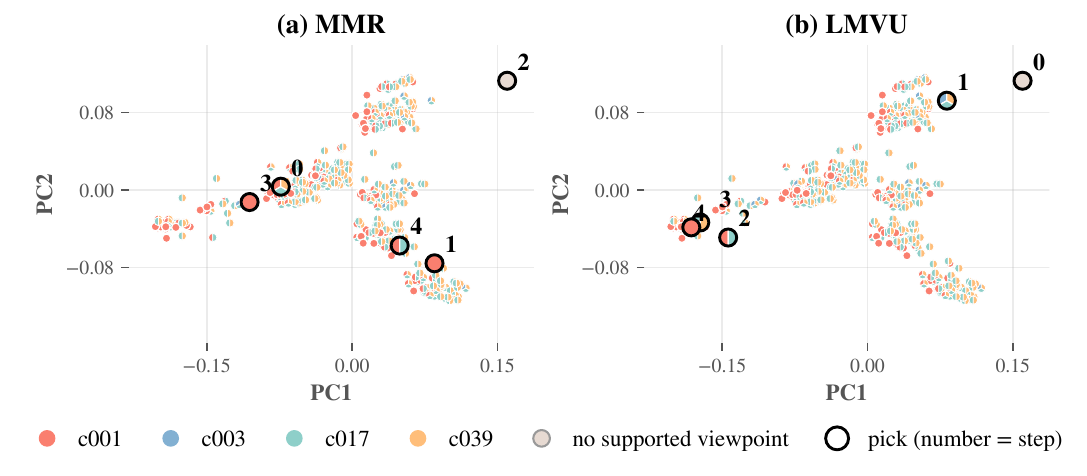}
\end{center}
\caption{Full selection pool of the WVS corruption query (508 items, budget 1024) in its first two principal components. Colors show gold-supported viewpoints (pie markers: several viewpoints; gold for display only); numbered rings give pick order (Table~\ref{tab:case_wvs}). (a)~MMR covers c001/c017/c039 with its first pick, and three of its four remaining picks are further c001 carriers that add no new viewpoint; (b)~LMVU's picks span the pool---from the no-viewpoint item both methods select at one end to one of the sixteen c003 carriers at step~1---and reach full coverage at step~2.}
\label{fig:viewspace}
\end{figure}

Figure~\ref{fig:viewspace} shows the geometry behind Table~\ref{tab:case_wvs}. Items supporting the pervasive-corruption viewpoint c001 dominate the 508-item pool (223 items, 89 of them supporting only c001), while the corruption-free viewpoint c003 is carried by sixteen items. MMR covers c001/c017/c039 with its first pick, but three of its four remaining picks are further c001 carriers, so 673 of its 891 tokens revisit already-covered viewpoints; LMVU's picks fan out across the pool, and its step-1 pick is one of the sixteen c003 carriers. Both allocators also select the same query-relevant item that supports no annotated viewpoint (respondent 0008), a relevance-driven choice the two methods share.

\subsection{ESS10 held-out: non-voting query}
\label{app:case_ess}

The ESS10 query asks whether the respondent voted in the last national election; the three annotated viewpoints are V1 = ``No'', V2 = ``Yes'', V3 = ``Not eligible to vote''. In Table~\ref{tab:case_ess}, all 12 MMR picks support the same answer (RVC $0.333$, 962 tokens), while LMVU covers all three viewpoints in its first three picks---V3 (``Not eligible to vote'') at step~0, V1 (``No'') at step~1, V2 (``Yes'') at step~2---then makes nine redundant picks (RVC $1.000$, 982 tokens).

\begin{table}[t]
\caption{Pick sequences for the ESS10 held-out vote query at budget 1024 (seed 17; the coverage-relevant first three picks are identical across the three training seeds): (a)~MMR, (b)~LMVU. Pick: anonymized respondent (country, gender, age group); remaining columns as in Table~\ref{tab:case_wvs}; consecutive picks that newly cover nothing are grouped into one row (step range, pick count, tokens summed). Tags: V1 = ``No'', V2 = ``Yes'', V3 = ``Not eligible to vote''. RVC-column shading encodes magnitude; darker cells indicate higher cumulative viewpoint coverage.}
\label{tab:case_ess}
\begin{center}
\small
\setlength{\tabcolsep}{1.5pt}
\begin{minipage}[t]{0.495\linewidth}
\centering
\textbf{(a) MMR}\par\vspace{2pt}
\begin{tabular}{rlrcc}
\toprule
Step & Pick & Tokens & New & RVC \\
\midrule
0 & GB-79894 (M, 60+) & 89 & V2 & \cellcolor{shadeblue!4}0.333 \\
1--11 & 11 picks & 873 & --- & \cellcolor{shadeblue!4}0.333 \\
\bottomrule
\end{tabular}
\end{minipage}\hfill
\begin{minipage}[t]{0.495\linewidth}
\centering
\textbf{(b) LMVU}\par\vspace{2pt}
\begin{tabular}{rlrcc}
\toprule
Step & Pick & Tokens & New & RVC \\
\midrule
0 & FR-25038 (M, 30-44) & 94 & V3 & \cellcolor{shadeblue!4}0.333 \\
1 & MK-27688 (F, 60+) & 86 & V1 & \cellcolor{shadeblue!16}0.667 \\
2 & PT-77928 (M, 60+) & 71 & V2 & \cellcolor{shadeblue!28}1.000 \\
3--11 & 9 picks & 731 & --- & \cellcolor{shadeblue!28}1.000 \\
\bottomrule
\end{tabular}
\end{minipage}
\end{center}
\end{table}

\subsection{A development failure state}

For balance, one development-only state (WVS political-participation query, step~0, dense trajectory family) where M1 underperforms, the MMR-style choice (respondent 1017) realizes a marginal gain of 7 viewpoints, while M1's choice (respondent 0517) realizes 4 against an oracle best of 9. At step~0 the selection state is empty, so the state-dependent features carry no signal and M1's ranking error is a pure scoring error. Such states are tallied in Appendix~\ref{app:failures}; they are most frequent on ArgKP, where they are inconsequential, and rarest relative to opportunity on WVS.